\pdfoutput=1
\documentclass[aoas,preprint]{imsart}

\arxiv{arXiv:1701.08140}

\RequirePackage[OT1]{fontenc}
\RequirePackage{amsthm,amsmath}
\RequirePackage{natbib}
\RequirePackage[colorlinks,citecolor=blue,urlcolor=blue]{hyperref}

\usepackage{algorithm, algorithmicx, algpseudocode}
\usepackage{caption}
\usepackage{subcaption}
\usepackage[utf8]{inputenc}
\usepackage{amsfonts}
\usepackage{mathrsfs}
\usepackage{tikz}
\usetikzlibrary{shapes,snakes}
\usepackage{float}
\usepackage{graphicx}
\usepackage{array}
\usepackage{booktabs}
\usepackage{url}

\startlocaldefs
\numberwithin{equation}{section}
\theoremstyle{plain}

\endlocaldefs

\newtheorem{lemma}{Lemma}
\theoremstyle{remark}
\newtheorem{remark}{Remark}

\newcommand{\real}{\mathbb{R}}                    % Reals
\newcommand{\prob}{\mathbb{P}}                    % Probability
\newcommand{\Tr}[1]{\text{Tr}\left(#1\right)}     % Trace 

\DeclareMathOperator*{\argmin}{arg\,min}

\DeclareMathOperator*{\prox}{prox}

\usepackage{lmodern}
\newcommand{\ADMM}{\text{\fontsize{4}{4}\selectfont ADMM}}
\newcommand{\PROX}{\text{\fontsize{4}{4}\selectfont PROX}}

\newcommand{\counter}{l}

\newcommand{\Bset}{\mathcal{B}}

\algnewcommand\algorithmicinput{\textbf{Input:}}
\algnewcommand\Input{\item[\algorithmicinput]}
\algnewcommand\algorithmiciterate{\textbf{Iterate:}}
\algnewcommand\Iterate{\item[\algorithmiciterate]}
\algnewcommand\algorithmicinitialize{\textbf{Initialize:}}
\algnewcommand\Initialize{\item[\algorithmicinitialize]}
\algnewcommand\algorithmicoutput{\textbf{Output:}}
\algnewcommand\Output{\item[\algorithmicoutput]}

\begin{document}

\begin{frontmatter}
\title{Network classification with applications to brain connectomics}
\runtitle{Network classification}

\begin{aug}
\author{\fnms{Jes\'us D.} \snm{Arroyo-Reli\'on}\thanksref{m1}\ead[label=e1]{jarroyor@umich.edu}},
\author{\fnms{Daniel} \snm{Kessler}\thanksref{m2}\ead[label=e2]{kesslerd@umich.edu}}
\author{\fnms{Elizaveta} \snm{Levina}\thanksref{m2}
\ead[label=e3]{elevina@umich.edu}}
\and
\author{\fnms{Stephan F.} \snm{Taylor}\thanksref{m2}\ead[label=e4]{sftaylor@umich.edu}}

\runauthor{J.D. Arroyo Reli\'on et al.}

\affiliation{Johns Hopkins University\thanksmark{m1} and University of Michigan\thanksmark{m2}}

\address{J. D. Arroyo Reli\'on\\
Center for imaging science\\
Johns Hopkins University\\
Baltimore, Maryland 21218-2682\\
USA\\
\printead{e1}}

\address{D. Kessler\\
E. Levina\\
Department of Statistics\\
University of Michigan\\
Ann Arbor, Michigan 48109–1107\\
USA\\
\printead{e2}\\
\phantom{E-mail:\ }
\printead*{e3}}

\address{S. F. Taylor\\
Department of Psychiatry\\
University of Michigan\\
Ann Arbor, Michigan 48109–1107\\
USA\\
\printead{e4}}

\end{aug}

\begin{abstract}

While statistical analysis of a single network has received a lot of attention in recent years, with a focus on social networks, analysis of a sample of networks presents its own challenges which require a different set of analytic tools.  Here we study the problem of classification of networks with labeled nodes, motivated by applications in neuroimaging.   Brain networks are constructed from imaging data to represent functional connectivity between regions of the brain, and previous work has shown the potential of such networks to distinguish between various brain disorders, giving rise to a network classification problem.   Existing approaches tend to either treat all edge weights as a long vector, ignoring the network structure, or focus on graph topology as represented by summary measures while ignoring the edge weights.  Our goal is to design a classification method that uses both the individual edge information and the network structure of the data in a computationally efficient way, and that can produce a parsimonious and interpretable representation of differences in brain connectivity patterns between classes. We propose a graph classification method that uses edge weights as predictors but incorporates the network nature of the data via penalties that promote sparsity in the number of nodes, in addition to the usual sparsity penalties that encourage selection of edges.    We implement the method via efficient convex optimization  and provide a detailed analysis of data from two fMRI studies of schizophrenia.  

\end{abstract}

%\begin{keyword}[class=MSC]
%\kwd[Primary ]{60K35}
%\kwd{60K35}
%\kwd[; secondary ]{60K35}
%\end{keyword}

\begin{keyword}
\kwd{graph classification, high-dimensional data, variable selection, fMRI data}
\end{keyword}

\end{frontmatter}

\section{Introduction\label{sec:intro}}
Network data analysis has received a lot of attention in recent literature, especially unsupervised analysis of a single network which is thought of as generated from an exchangeable random graph model, for example \cite{bickel2009nonparametric,le2015sparse,zhang2016minimax,gao2015achieving}  and many others.   This setting is applicable to a number of real life scenarios, such as social networks, but there are situations where the network nodes are labeled and therefore not exchangeable, and/or more than one network is available for analysis, which have received relatively less attention.  Here we focus on the setting motivated by brain connectomics studies, where a sample of networks is available from multiple populations of interest (for example, mentally ill patients and healthy controls).    In this setting, each unit in the population (e.g., a patient) is represented by their own network, and the nodes (brain regions of interest) are labeled and shared across all networks through a registration process that maps all individual brains onto a common atlas.   There are many classical statistical inference questions one can ask in this setting, for example, how to compare different populations \citep{tang2017nonparametric, tang2017semiparametric}.   The question we focus on in this paper is a classification problem: given a training sample of networks with labeled nodes drawn from multiple classes, the goal is to learn the rules for predicting the class of a given network, and just as importantly, interpret these rules.

Network methods are a popular tool in the neuroscience literature \citep{bullmore2009complex,bullmore2011brain}. A brain network represents connectivity between different locations of an individual's brain. How connectivity is defined varies with the type of imaging technology used and the conditions under which data were collected.   In this paper, we focus on functional connectivity, which is a measure of statistical association between each pair of locations in the brain, constructed from functional magnetic resonance imaging (fMRI) data, although the methods we develop are applicable to any sample of weighted networks with labeled nodes. In fMRI studies, BOLD (blood oxygen-level dependent) signal, a known correlate of underlying neural activity, is measured at a sequence of time points at many spatial locations in the brain, known as voxels,  resulting in a 4-dimensional data array, with three spatial dimensions and a time index. Brain networks constructed from fMRI data have been successfully used for various tasks, such as differentiating between certain illnesses, or between types of external stimuli  \citep{bullmore2009complex}, and contain enough information to identify individual subjects \citep{finn2015functional}.    Extensive statistical literature has focused on the analysis of raw fMRI data \citep{lindquist2008statistical,zhou2013tensor,zhang2016spatiotemporal}, usually aiming to characterize brain activation patterns obtained from task-based fMRI experiments. In this paper, we focus on resting-state fMRI data, where no particular task is performed  and subjects are free to think about anything they want.  Thus registering the time dimension across different subjects is not possible.    The connectivity network approach, which averages over the time dimension in computing a measure of dependence between different voxels, is thus a natural choice, and has been widely used with multiple types of neuroimaging data.

\begin{figure}
	\centering
	\includegraphics[width=0.95\textwidth]{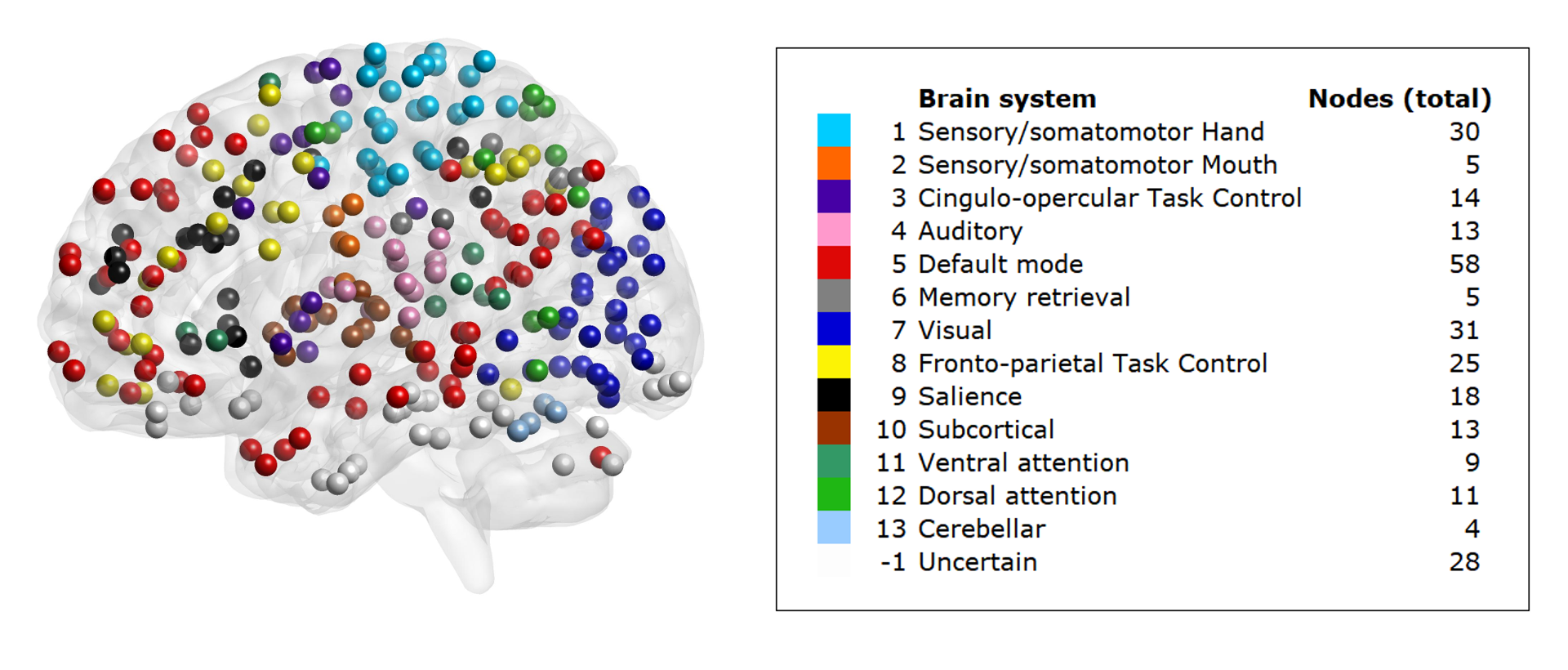}
	\caption{Regions of interest (ROIs) defined by \cite{power2011functional}, colored by brain systems, and the total number of nodes in each system. \label{fig:powerparcellation}} 
\end{figure}

Two different datasets are analyzed in this paper, both of resting state fMRI studies on schizophrenic patients and healthy controls.  One dataset, COBRE (54 schizophrenics and 70 controls), is publicly available \citep{Aine2017};  another, which we will refer to as UMich data (39 schizophrenics and 40 controls), was collected internally in the last author's lab.     Having two datasets on the same disease allows us to cross-check models trained on one of them for classification on the other to check the robustness of our approach. The raw data arrays undergo pre-processing and registration steps, discussed in detail in the Appendix \ref{appendix:data}, along with additional details on data collection.     To construct a brain network from fMRI measurements, a set of nodes is chosen, typically corresponding to regions of interests (ROIs) from some predefined parcellation. In our analysis we use the parcellation of \cite{power2011functional} (see Figure \ref{fig:powerparcellation}), which consists of 264 ROIs divided into 14 functional brain systems (in the COBRE data, node 75 is missing).  A connectivity measure is then computed for every pair of nodes, resulting in an adjacency matrix of size $264\times 264$. Many choices of connectivity measures are available  \citep{smith2013functional};  perhaps the most commonly used one is the Pearson correlation coefficient between locations, computed by averaging over the time dimension.   It has been argued that partial correlations are a better measure of connectivity \citep{varoquaux2013learning, narayan2015two}, but the choice depends on the final goal of analysis. In this paper we follow the vast majority of the connectomics literature and measure connectivity on each individual by using marginal correlations between the corresponding time series (see Figure \ref{fig:brain-adjacency}).  The correlations are then further rank-transformed and standardized;  see Appendix \ref{appendix:data} for details.    These transformations are intended to deal with subject-to-subject variability and the global signal regression issue \citep{gotts2013perils}, and although they lose some information, we observed that on our datasets classification based on standardized ranks of marginal correlations outperformed classification based on other connectivity measures, such as marginal correlations. The methods we develop here are applicable to networks that encode any type of connectivity measure.
\begin{figure}
	\centering
	\includegraphics[width=0.6\textwidth]{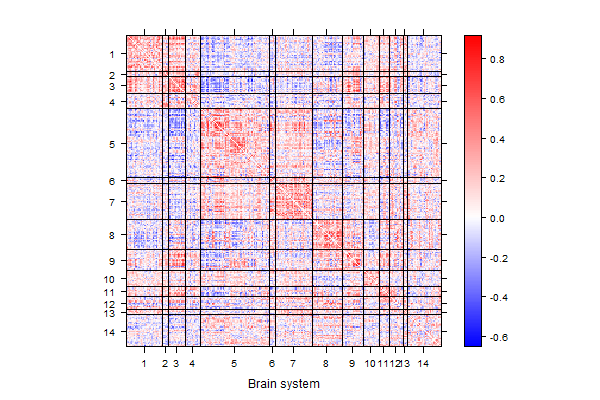}
	\caption{Brain network from one of the subjects, showing the value of the Fisher z-transformed correlations between the nodes, with the 264 nodes grouped into 14 brain systems.\label{fig:brain-adjacency}}
\end{figure}

The problem of graph classification has been studied previously in other contexts, with a substantial literature motivated by the problem of classification of chemical compounds \citep{srinivasan1996theories, helma2001predictive}, where graphs represent the compound's molecular structure. This setting is very different, with small networks of about 20 nodes on average, binary or categorical edges recorded with no noise, and different nodes corresponding to different networks, \citep{ketkar2009empirical}.   
Classification methods for chemical compounds are usually based on finding certain discriminative patterns in the graphs, like subgraphs or paths  \citep{inokuchi2000apriori, gonzalez2002graph}, and using them as features for training a standard classification method  \citep{deshpande2005frequent,kudo2004application,fei2010boosting}.  Computationally, finding these patterns is only possible on small binary networks.   

Another type of methods are based on graph kernels \citep{gartner2003graph,vishwanathan2010graph}, which  define a similarity measure between two networks. These kernels combined with support vector machines (SVMs) have been successfully used on small networks \citep{kashima2003marginalized,borgwardt2005protein}, but the curse of dimensionality makes local kernel methods unsuitable for large scale networks \citep{bengio2005non}.    On our datasets, graph kernel methods did not perform better than random guessing.

In the context of classifying  large-scale brain networks, two main approaches have been followed.  One approach is to reduce the network to its global summary measures such as the average degree, clustering coefficient, or average path length \citep{bullmore2009complex},  and use those measures as features for training a classification method. Previous studies have reported significant differences on some of these network measures for groups of patients with certain brain diseases compared with healthy controls 
\citep{supekar2008network, liu2008disrupted}, suggesting their usefulness as diagnostic biomarkers. However, global summary statistics collapse all local network information, which can harm the accuracy of classification and do not allow to identify local differences.   In our data analysis,  a method based on the network measures suggested in \cite{prasad2015brain}  performed poorly for classification  (see Section \ref{sec:data}).

An alternative approach to classification of large networks is to treat edge weights as a ``bag of features'', vectorizing the unique elements of the adjacency matrix and ignoring the network nature of the data.    This approach can leverage many existing classification methods for vectors, and provides an interpretation at the edge level if variable selection is applied \citep{richiardi2011decoding, craddock2009disease, zhang2012pattern}.  Spatial correlation between edges connecting neighboring nodes can be incorporated \citep{watanabe2014disease,scott2015false}, although the effectiveness of this regularization will depend on the parcellation used to define nodes. Alternatively, an individual test can be used for each edge to find significant differences between two populations, with a multiple testing correction and without constructing a classifier at all \citep{narayan2015two}. While these methods can deliver good predictions, their interpretability is limited to individual edge selection, which is less scientifically interesting than identifying differentiating nodes or regions, and they cannot account for network structure.

Taking the network structure into account can have benefits for both testing and classification settings. Some methods perform inference over groups of edges based on the community assignments of the nodes to which they are incident. For example, \cite{sripada2014volitional,sripada2014disrupted} introduced \emph{Network Contingency Analysis} which begins with massive univariate testing at each edge, and then counts the number of superthreshold connections in each \emph{cell}, a group of edges that connect nodes in two functional systems.  Nonparametric methods are then used to conduct inference on the count statistic for each cell, with multiple comparison correction for inference at the cell level.    Power can be improved  by applying a network-based  multiple testing dependence correction \citep{zalesky2010network}.  For classification,  better interpretability and potentially accuracy can be obtained if we focus on understanding which brain regions or interactions between them are responsible for the differences.     In somewhat related work, \cite{vogelstein2013graph} proposed to look for a minimal set of nodes which best explains the difference, though that requires solving a combinatorial problem.   Hypothesis testing on a type of graph average has also been proposed \citep{ginestet2017hypothesis}. Bayesian nonparametrics approaches for modeling populations of networks allow to test for local edge differences between the groups \citep{durante2017bayesian}, but are computationally feasible only for small networks.  

Our goal in this paper is to develop a high-dimensional network classifier  that uses all the individual edge weights but also respects the network structure of the data and  produces more interpretable results.   To achieve this goal, we use structured sparsity penalties to incorporate the network information by penalizing both the number of edges and the number of nodes selected.  Although our main application here is classification of brain connectivity networks, our methods are applicable to any weighted graphs with labeled nodes, and to general prediction problems, not just classification.  

The rest of this paper is organized as follows.   In Section \ref{sec:model}, we introduce our classifier and the structured penalties.  In Section \ref{sec:opt} we show how to efficiently solve the resulting convex optimization problem by a proximal algorithm, each step of which is a further optimization problem which we solve by the alternating direction method of multipliers (ADMM). 
The performance of our method is evaluated and compared with other methods using simulations in 
Section \ref{sec:simulations}.
In Section \ref{sec:data}, we analyze two brain connectivity datasets, each containing schizophrenic patients and healthy controls, and show that our regularization framework leads to state-of-the-art  accuracy while providing interpretable results, some of which are consistent with previous findings and some are new.  We conclude with a brief discussion in Section \ref{sec:disc}.

%%%%%%%%%%%%%%%%%%%%%%%%%%%%%%%%%%%%%%%%%%%%%%%%%%%%%%%%%%%%%%%%%%%%%%%%%%%%%%%%%%%%%%%%%%%%%%%%%%%%%%

\section{A framework for node selection in graph classification\label{sec:model}}

\subsection{A penalized graph classification approach\label{sec:penalty}}
We start from setting up notation.  All graphs we consider are defined on the same set of $N$ labeled nodes. % Define $\Graph=(\Vertex,\Edge)$ to be a complete graph over this set, that is, every pair of nodes  $V_i$ and $V_j$ on $\Vertex$ have an edge $E_{ij}=(V_i,V_j)$ on $\Edge$ that links them. This graph $\Graph$ defines the overall structure of the population of graphs.   -- I can't see why this is necessary. 
A graph can be represented with its adjacency matrix $A\in\real^{N\times N}$.   We focus on graphs that are undirected ($A_{ij} = A_{ji}$) and contain no self-loops ($A_{ii} = 0$).   These assumptions are not required for the derivations below, but they match the neuroimaging setting and simplify notation. Our goal is  predicting a class label $Y$ from the graph adjacency matrix $A$;  in this paper we focus on the binary classification problem where $Y$ takes values $\{-1,1\}$, although extensions from binary to multi-class classification or real-valued responses are straightforward.   Throughout this paper, we use $\|\cdot\|_p$ to denote the entry-wise $\ell_p$ norm, i.e., for a matrix $A\in\real^{m\times n}$, $\|A\|_p=\left(\sum_{i=1}^{m}\sum_{j=1}^n|A_{ij}|^p\right)^{1/p}$.

A standard general approach  is to construct a linear classifier, which predicts the response $Y$ from a linear combination of the elements of $A$, $ \langle A,B\rangle = \Tr{B^TA}$, where we arrange the coefficients in a matrix $B\in \real^{N\times N}$ to emphasize the network nature of the predictors.    We focus on linear classifiers here because variable selection is at least as important as prediction itself in the neuroimaging application, and setting coefficients to 0 is a natural way to achieve this.    The coefficients are typically estimated from training data by minimizing an objective consisting of a loss function plus a penalty.  The penalty can be used to regularize the problem to  make the estimator well-defined in high-dimensional problems, to select important predictors, and to impose structure, and many such penalties have been proposed, starting from the lasso \citep{tibshirani1996regression}.  Our focus is on designing a classifier in this framework that respects and utilizes the network nature of the predictors.    In brain networks in particular, neuroscientists believe that edges are organized in subnetworks, also called brain systems \citep{power2011functional}, that carry out specific functions, and 
certain subnetworks are important for prediction \citep{bullmore2009complex}, although different studies tend to implicate different regions \citep{fornito2012schizophrenia}.   Thus we aim to find nodes or subnetworks with good discriminative power, and hence select both the most informative nodes and edges.

Although the methods we develop here can be used on small networks,  our main focus is on the more challenging case of medium to large brain networks. In brain connectivity studies dealing with multiple subjects, while raw images may contain hundreds of thousands of voxels, they are commonly down-sampled according to a parcellation scheme with a coarser resolution, usually resulting in networks with hundreds or thousands of nodes representing ROIs (see for example \cite{kiar2018high}). This coarser resolution is essential for registration, as aligning different brains at a high resolution is much harder, but it still results in hundreds of thousands or millions of edges which serve as predictor variables.   Given the typical data sizes in this area of application, we focus on methods based on convex formulations, which allow for efficient and scalable implementations with convergence guarantees.

Let $\{(A^{(1)},Y_1),\ldots,(A^{(n)},Y_n)\}$ be the training sample of undirected adjacency matrices with their class labels, 
and let  $Y=(Y_1,\ldots,Y_n)$.  
A  generic linear classifier described above is computed by finding the coefficients $B$ defined by 
\begin{equation}
	\hat{B} =  \argmin _{B\in\Bset} \left\{ \ell(B )  + \Omega(B)\right\}, \label{eq:defin_est}
\end{equation}
where $\Bset = \left\{B\in\real^{N\times N}:B=B^T, \text{diag}(B) = 0\right\}$, $\Omega$ is a penalty, and 
$$\ell(B) = \frac{1}{n}\sum_{k=1}^n \tilde\ell(Y_k, A^{(k)}; B ) $$
is a loss function evaluated on the training data.  Our methodology can accommodate different choices of loss functions that can extend beyond classification problems (e.g., least squares or generalized linear models). The optimization algorithm presented in Section \ref{sec:opt} can work with any convex and continuously differentiable loss function, and further assumptions are required for consistency (see Section \ref{sec:theory}). In this paper, for the purpose of classification we use the logistic loss function in the simulations and data analysis, which is defined as 
\begin{equation*}
	\tilde{\ell}\left(Y_k, A^{(k)};B,b\right) = \log\left\{1+\exp\left(-Y_k \langle A^{(k)}, B\rangle+b\right)\right\} \ . 
\end{equation*}
The threshold $b$ is an additional parameter to be estimated.

To capture structural assumptions on important predictive edges, we focus on convex structured sparsity penalties \citep{bach2012structured} that  encourage a small number of active nodes, by which we mean nodes attached to at least one edge with a non-zero coefficient. One approach to finding a set of such nodes was proposed by \cite{vogelstein2013graph}, who called it a signal-subgraph, and proposed finding the minimal set of nodes (called signal vertices) which together are incident to all selected edges (but not every node connected to a selected edge is a signal vertex).    Finding this set is a combinatorial optimization problem, and the set is not always uniquely defined.    Instead, we focus on convex formulations that allow for efficient computation and encourage small active node sets indirectly.

Other convex penalties have been used for fMRI data as a way to enforce spatial smoothness in the solution  \citep{grosenick2013interpretable,xin2014efficient,hu2015local}. These methods assume that voxels are equally spaced in the brain, and neighboring voxels are highly correlated. In particular, \cite{watanabe2014disease} proposed penalties for brain network classification using these spatial assumptions.  Here, instead of enforcing a spatial regularization directly, we aim for a regularization that can be applied to any type of network data, and in particular to brain networks with coarse and/or uneven parcellations where enforcing spatial smoothness may not work as well.   In any case, the flexibility of convex optimization algorithms allows one to easily incorporate additional spatially-informed penalties if needed.

\subsection{Selecting nodes and edges through group lasso}

To reflect the network structure of the predictors, we use a penalty that promotes a sparse classifier not only in the number of edges used, but also in the number of nodes.   The group lasso penalty  \citep{yuan2006model} is designed to eliminate a group of variables simultaneously.   Here we penalize the number of active nodes by treating all edges connected to one node as a group. Then eliminating this group (a row of coefficients in the matrix $B$) is equivalent to de-activating a node. The group penalty is defined as
\begin{equation}
	\Omega_{\lambda, \rho}(B)= \lambda\left(\sum_{i=1}^N\|B_{(i)}\|_2 + \rho\|B\|_1\right), \label{eq:GL_penalty}
\end{equation}
where $B_{(i)}$ denotes the vector of edge weights incident to the $i$-th node (or equivalently, the $i$-th row or column of $B$), and $\lambda, \rho\geq 0$ are tuning parameters.  Note that the constraint $B=B^T$ makes the groups overlap, since a coefficient $B_{ij}$  belongs to groups associated with the nodes $i$ and $j$, and therefore, the edge between nodes $i$ and $j$ would be selected only if both nodes are activated.    The second term in the penalty $\rho\|B\|_1$ acts as the usual lasso penalty to promote sparsity inside the group \citep{friedman2010note}, allowing to select a subset of edges for an active node. Due to the overlap in the groups, this lasso penalty is usually necessary in order to produce sparse solution (see Proposition \ref{prop:error}). The constraint $\text{diag}(B) = 0$ in \eqref{eq:defin_est} is automatically enforced with this formulation. %On the other hand, due to the overlap of the groups, small values of $\rho$ will not produce any zeros in the coefficients of the solution (see Proposition \ref{prop:error}). This problem can also be solved by other formulations of overlapping group lasso, as explained in Remark \ref{remark:overlapping_union}. 

\begin{remark} \label{remark:overlapping_union}
	An alternative to the constraint in the problem \eqref{eq:defin_est} is to optimize over the set \[\tilde\Bset = \left\{B\in\real^{N\times N}, \text{diag}(B) = 0\right\}.\] Without the symmetry constraint and assuming undirected graphs, the penalty \eqref{eq:GL_penalty} is equivalent to the overlapping group lasso formulation of \cite{jacob2009group}. This formulation has some advantages. Since it gives group lasso without overlaps, the lasso penalty $\rho\|B\|_1$ is not required to obtain sparse solutions, and more efficient optimization algorithms exist for this case. This approach would loosely correspond to the idea of selecting signal nodes as in \cite{vogelstein2013graph}, in the sense that an edge can be selected if at least one of its nodes is selected, and the second node could be inactive. The downside is that each edge now corresponds to two different coefficients $B_{ij}$ and $B_{ji}$, the problem encountered by all variable selection methods that ignore symmetry, such as \cite{meinshausen2006high}. The standard solution for this problem, as suggested by \cite{jacob2009group}, is to take the average of the coefficients. 
	Intuitively, one would expect that the formulation using $\Bset$ would be better when the significant edges are incident to a small set of nodes, since both nodes have to be active for an edge to be selected, while using $\tilde \Bset$ may be better when for some nodes most of their edges are significant, creating ``significant hubs''.  Since in our application we are primarily looking for discriminative brain subnetworks, we focus on the symmetrically constrained formulation for the rest of the paper.  We also found that in practice this second formulation results in less accurate classifiers for the neuroimaging data discussed in Section \ref{sec:data}.
\end{remark}

\begin{remark} \label{remark:directedgraph} The analogue to \eqref{eq:GL_penalty} for directed graphs would assign coefficients $B_{ij}$ and $B_{ij}$ to the same group, resulting in the penalty 
	\begin{equation}
		\Psi_{\lambda, \rho}(B)= \lambda\left(\sum_{i=1}^N\sqrt{\sum_j\left(B_{ij}^2+B_{ji}^2\right)} + \rho\|B\|_1\right), \label{eq:GL_penalty_directed}
	\end{equation}
	where $B\in\tilde{\mathcal{B}}$. Alternatively, we can also use the formulation of Remark \ref{remark:overlapping_union}, by replicating the variables and estimating two matrices of coefficients, say $B^{(1)}$ and $B^{(2)}$, with the penalty 
	\begin{equation*}
		\tilde\Psi_{\lambda, \rho}(B^{(1)},B^{(2)}) = \lambda\left[\sum_{i=1}^N\sqrt{\sum_j\{(B^{(1)}_{ij})^2+(B^{(2)}_{ji})^2\}} + \rho(\|B^{(1)}\|_1 + \|B^{(2)}\|_1 )\right], \label{eq:GL_penalty_directed_union}
	\end{equation*}
	with $B^{(1)},B^{(2)}\in\tilde{\mathcal{B}}$, and set the coefficients matrix to $B=\left(B^{(1)} + B^{(2)}\right)/2$.   This formulation will again not directly select subnetworks as discussed in Remark \ref{remark:overlapping_union}. 
\end{remark}

Finally, for numerical stability we add an extra ridge penalty term $\frac{\gamma}{2}\|B\|_F^2=\frac{\gamma}{2}\text{Tr}\left(B^TB\right)$, with $\gamma$ a small and fixed constant. There are several benefits of combining ridge and lasso penalties (see for example \cite{zou2005regularization}).  The parameter $\gamma$ can be potentially considered as an additional tuning parameter, but  here we only use a small fixed constant $\gamma$ in order to avoid numerically degenerate solutions. In practice, the results are not sensitive to the exact value of $\gamma$.   

Putting everything together, to fit our graph classifier, we solve the problem 
\begin{align}
	(\hat{B},\hat{b}) =  \argmin_{B\in\mathcal{B},b\in\real} & \left\{ \frac{1}{n}\sum_{k=1}^n\log\left(1+\exp(-Y_k \langle B,A^{(k)}\rangle+b)\right) + \frac{\gamma}{2}\|B\|_F^2 +  \right. \label{eq:main_problem} \\
	 & \left.  \lambda\left( \sum_{i=1}^N\|B_{i}\|_2+\rho\|B\|_1\right)\right\} 
	\nonumber
\end{align}
for given values of $\lambda$ and $\rho$, which will be chosen by cross-validation.

%%%%%%%%%%%%%%%%%%%%%%%%%%%%%%%%%%%%%%%%%%%%%%%%%%%%%%%%%%%%%%%%%%%%%%%%%%%%%%%%%%%%%%%%%%%%%%%%%%%%%%%%%%
\section{The optimization algorithm\label{sec:opt}}

Our optimization algorithm to solve the problem \eqref{eq:main_problem} combines two common approaches to convex optimization: proximal algorithms and alternating direction method of multipliers (ADMM).  We use an accelerated version of the proximal algorithm \citep{beck2009fast} to solve the main problem \eqref{eq:main_problem}. In each step, we need to calculate a proximal operator, which is a further convex optimization problem solved with the ADMM algorithm.

The main optimization difficulty comes from the overlapping groups. Some algorithms have been proposed for this case, including a subgradient descent method \citep{duchi2009efficient}, which has a slow rate of convergence, or proximal algorithms based on smoothing the original problem \citep{yuan2011efficient,chen2012smoothing}. Although smoothing yields fast algorithms, it is not clear that the sparsity pattern is preserved with those approximations. We follow an approach similar to \cite{yuan2011efficient} and \cite{chen2012smoothing}, but solve the proximal operator for the penalty \eqref{eq:GL_penalty} directly using the ADMM method. This can potentially give a more accurate sparsity pattern, and the flexibility of the algorithm allows for additional penalties if desired, such as spatial smoothing similar to \cite{watanabe2014disease} (see Remark \ref{remark:optimization-other-penalties}).

The main problem \eqref{eq:defin_est} is solved with a proximal algorithm (see \cite{parikh2013proximal}). Recall that the proximal operator  for a function $f$ is defined as $\prox_f(v) = \argmin_x\{f(x)+\frac{1}{2}\|x-v\|_2^2\}$. Starting with an initial value $B^{(0)}\in\real^{N\times N}$, a proximal algorithm solves the optimization problem \eqref{eq:defin_est} 
by iteratively calculating the proximal operator of $\Omega=\Omega_{\lambda,\rho}$ for a descent direction of the differentiable loss function $\ell$.  We use an accelerated version of the algorithm \citep{beck2009fast}, which for each $k=2,\ldots,$ until convergence, performs the updates
\begin{align}
	W^{(k)} & =  B^{(k-1)} + \frac{k-1}{k+2}\left(B^{(k-1)} - B^{(k-2)}\right) \label{eq:proximalalg_direction}\\
	B^{(k)} & =  \prox_{t^{(k)}\Omega}\left\{W^{(k)} - t^{(k)}\nabla\ell(W^{(k)})\right\}\label{eq:proximal_signalapprox}\\
	& =  \argmin_{B\in\Bset} \left\{\frac{1}{2}\left\|B-\left(W^{(k)}-t^{(k)}\nabla \ell(W^{(k)})\right)\right\|_2^2 + t^{(k)}\Omega(B) \right\},\nonumber
\end{align}
where $\nabla\ell(W)\in\real^{N\times N}$ is the gradient of the loss function $\ell$ at $W$ and $\{t^{(k)}\}$ is a sequence of positive values. If $\nabla\ell$ is Lipschitz continuous, with $L$ its Lipschitz constant, the sequence of values $\ell(B^{(k)})+\Omega(B^{(k)})$ converges to the optimal value at rate $O(1/k^2)$ if $t^{(k)}\in[0,1/L)$.    The value of $t^{(k)}$ can be chosen using a backtracking search \citep{beck2009fast}, which decreases this value until the condition
\begin{equation}
	\ell(B^{(k)}) \leq \ell(W^{(k)}) + \left\langle\nabla\ell(W^{(k)}), B^{(k)}-W^{(k)} \right\rangle  + \frac{1}{2t^{(k)}}\|B^{(k)}-W^{(k)}\|_2^2
	\label{eq:backtrackingcondition}
\end{equation}
is satisfied. This procedure ensures that step sizes $\{t^{(k)}\}$ become smaller as the algorithm progresses, until $t^{(k)}<1/L$. In practice, $L$ might be large, which can make the algorithm slow to converge. It has been observed in other sparse high-dimensional problems that search strategies for $t^{(k)}$ which allow for $t^{(k)}  > 1/L$ when appropriate can actually speed up convergence  \citep{scheinberg2014fast, hastie2015statistical}. 
We use a strategy of this type, allowing $t^{(k)}$ to increase by a factor of $\alpha\geq 1$ if the relative improvement in the loss function on iteration $k$ becomes small. We observed that this strategy significantly reduces the number of iterations until convergence. The entire procedure is summarized in Algorithm \ref{alg:prox} on the Appendix~\ref{appendix:optimization}.

The logistic loss function of \eqref{eq:main_problem} has an extra parameter $b$. Rather than including it as an unpenalized coefficient for a constant covariate, we use block coordinate descent and solve for $b$ separately. This is convenient because the threshold $b$ and the matrix of coefficients $B$ may not be on the same scale. Thus, $b$ can be updated by solving $b^{(k+1)}  = \argmin_{b\in\real} \ell(B^{(k)},b)$, which is easy to compute via Newton's method.

The proximal algorithm requires solving the proximal operator \eqref{eq:proximal_signalapprox}, which has no closed form solution for the penalty \eqref{eq:GL_penalty} under the symmetry constraint. Strategies based on smoothing this penalty have been proposed \citep{yuan2011efficient,chen2012smoothing}. However, to allow for variable selection which results from non-differentiability of the penalty, we aim to solve the proximal operator directly using ADMM  (see \cite{boyd2011distributed} for a review). Note that if the symmetric constraint is relaxed as in Remark \ref{remark:overlapping_union}, the proximal operator has a closed form solution (see Remark \ref{remark:optimization-union}).

The ADMM works by introducing additional constraints and performing coordinate descent in the corresponding augmented Lagrangian function. Setting $Z=W^{(k)}-t^{(k)}\nabla \ell(W^{(k)})$ and $t=t^{(k)}$, and introducing the variables $Q,R\in\real^{N\times N}$, we can formulate \eqref{eq:proximal_signalapprox} as a convex optimization problem
\begin{equation}
	\begin{aligned}
		& \min_{\tilde B,Q,R} & &\frac{1}{2}\|\tilde B-Z\|_2^2 + t\lambda\left(\sum_{i=1}^N\|Q_{(i)}\|_2 + \rho\|R\|_1\right)\\
		& \text{subject to}
		& & \tilde B = Q, \quad\tilde B = R, \quad\tilde B = \tilde B^T,\quad \text{diag}(\tilde B) = 0.
	\end{aligned} \label{eq:ADMM_general}
\end{equation}
The ADMM algorithm introduces the multipliers $U,V\in\real^{N\times N}$ and a penalty parameter $\mu>0$ to perform gradient descent on the Lagrangian of \eqref{eq:ADMM_general}, given by $\mathcal{L}_\mu=\mathcal{L}_\mu(\tilde{B},Q,R,U,V)$ as
\begin{align}
	\mathcal{L}_\mu  = &\frac{1}{2}\|\tilde{B}-Z\|_2^2+ t\lambda\left(\sum_{i=1}^N\|Q_{(i)}\|_2 + \rho\|R\|_1\right) +  \langle U, \tilde{B}-Q\rangle+\label{eq:lagrangian}\\
	&  \langle V, \tilde{B}-R\rangle  \frac{\mu}{2}\left(\|\tilde{B}-Q\|_2^2+\|\tilde{B}-R\|_2^2 + \|\tilde{B}-\tilde{B}^T\|_2^2\right).\nonumber
\end{align}
The value $\mu$ controls the gap between dual and primal feasibility. In practice, we observed that setting $\mu = 0.1$ gives a good balance between primal and dual feasibility, although other self-tuning methods are available \citep{parikh2013proximal}. This function is optimized by coordinate descent, with each variable updated to minimize the value of $\mathcal{L}_\mu$ while all the other variables are fixed. This update has a closed form; the detailed steps of the ADMM are shown in Algorithm \ref{alg:admm} on the Appendix~\ref{appendix:optimization}. These steps are performed until the algorithm converges within tolerance $\epsilon^\ADMM>0$. %, where the convergence is measured according to the dual and primal residuals, detailed in Algorithm \ref{alg:admm}. 
Note that ADMM will be performed in each iteration of the algorithm to solve \eqref{eq:main_problem} and thus tolerance $\epsilon^\ADMM$ can be decreased as the algorithm progresses. On the other hand, performing only one iteration of algorithm \eqref{alg:admm} gives a similar algorithm to the one of \cite{chen2012smoothing}.

\begin{remark}\label{remark:optimization-other-penalties}
	The ADMM  makes it very easy to incorporate additional penalties.  If $\Psi$ is a new penalty, we can rewrite \eqref{eq:ADMM_general} by introducing an additional parameter $\tilde{Q}$ so it becomes
	\begin{equation*}
		\begin{aligned}
			& \min_{\tilde B,Q,\tilde{Q},R} & &\frac{1}{2}\|\tilde B-Z^{(k)}\|_2^2 + t\lambda\left(\sum_{i=1}^N\|Q_{(i)}\|_2 + \rho\|R\|_1\right) + t\Psi(\tilde{Q})\\
			& \text{subject to}
			& & \tilde B = Q,\quad  \tilde B = \tilde Q, \quad\tilde B = R, \quad\tilde B = \tilde B^T,\quad \text{diag}(\tilde B) = 0.
		\end{aligned} 
	\end{equation*}
	We can obtain the Lagrangian formulation \eqref{eq:lagrangian} in a similar manner, and include new parameters in the ADMM updates, which can be performed efficiently as long as the proximal operator of $\Psi$ has a closed form solution. This is in fact the case for some other penalties of interest, such as the GraphNet penalty \citep{grosenick2013interpretable,watanabe2014disease}, which can incorporate spatial location information.  
\end{remark}

\begin{remark}\label{remark:optimization-union}
	The alternative formulation for the graph penalty given in Remark \ref{remark:overlapping_union} corresponds to standard sparse group lasso \citep{friedman2010note}. In particular, we can still employ the proximal algorithms \eqref{eq:proximalalg_direction} and \eqref{eq:proximal_signalapprox}, but instead optimize over the set $\tilde{\Bset}$. Without the symmetric constraint on $B$, the overlap in the group lasso penalty disappears, and this vastly simplifies the problem. Using Theorem 1 of \cite{yuan2011efficient}, the update for $B^{(k)}$ has a closed form solution given by
	\begin{align}
		Y^{(k)} & =  B^{(k-1)} + \frac{k-2}{k}\left(B^{(k-1)} - B^{(k-2)}\right) \label{eq:proximalalg_direction_union}\\
		Z^{(k)}_{ij} & = \left(1- \frac{\lambda \rho}{\left|Y^{(k)}_{ij}-t_k\nabla_{ij}\ell(Y^{(k)})\right\|_2}\right)_+\left(Y_{ij}^{(k)}-t_k\nabla_{ij} \ell(Y^{(k)})\right)\\
		B^{(k)}_{(i)} & =  \left(1-\frac{\lambda}{\left\|Z^{(k)}_{(i)}\right\|_2}\right)_+\left(Z^{(k)}_{(i)}\right),\quad\quad i\in[N].\label{eq:proximal_union}
	\end{align}
\end{remark}

%%%%%%%%%%%%%%%%%%%%%%%%%%%%%%%%%%%%%%%%%%%%%%%%%%

\section{Theory\label{sec:theory}}

In this section, we show that the solution of the penalized problem \eqref{eq:GL_penalty} can recover the correct subgraph corresponding to the set of non-zero coefficients, and give its rates of convergence.  The theory is  a consequence of the results of \cite{lee2015model} for establishing model selection consistency of regularized M-estimators under geometric  decomposability (see Appendix for details). We present explicit conditions for our penalty to work well, which depend on the data as well as the tuning parameters.

Let $B^\star\subset\real^{N\times N}$ be the unknown parameter we seek to estimate, and we assume there is a set of active nodes $\mathcal{G}\subset[N]$ with $|\mathcal{G}|=G$, so that  $B^\star_{ij}=0$ if $i\in\mathcal{G}^C$ or $j\in\mathcal{G}^C$. We allow some edge weights inside the subgraph defined by $\mathcal{G}$ to be zero, but we focus on whether the set $\mathcal{G}$ is correctly estimated by the set $\hat{\mathcal{G}}$ of active nodes in $\hat{B}$.   Denote by $\mathcal{M} \subseteq\real^{N\times N}$ the set of matrices where the only non-zero coefficients appear in the active subgraph, that is,
\begin{equation}
	\mathcal{M}=\left\{B\in\real^{N\times N}\left| B_{ij}=0\text{ for all }i\in\mathcal{G}^C\text{ or }j\in\mathcal{G}^C, B=B^T\right.\right\}\label{eq:m_activeset}
\end{equation}
There are two main assumptions on the loss function $\ell$ required for consistent  selection in high-dimensional models \citep{lee2015model}. The first assumption is on the convexity of the loss function around $B^\star$, while the second assumption bounds the size of the entries in the loss Hessian between the variables in the active subgraph and the rest. Let the loss Hessian $\nabla^2\ell(B^\star)\in\real^{N\times N}\otimes\real^{N\times N}$ be defined by 
\[\nabla^2_{(i,j),(k,l)}\ell(B) = \frac{\partial^2\ell(B)}{\partial B_{ij}\partial B_{kl}},\]
and define the matrix $H_{(i,j),\mathcal{G}}\in\real^{G\times G}$ with $(i,j)\in\left(\mathcal{G}\times\mathcal{G}\right)^C$ such that 
%\begin{equation}
%\left(H_{(i,j),\mathcal{G}}\right)_{k,l} = \Tr{\nabla^2_{(i,j),(\mathcal{G},\mathcal{G})}\ell(B^\star) \left(\nabla^2_{(\mathcal{G},\mathcal{G}),(\mathcal{G},\mathcal{G})}\ell(B^\star)_{(k,\cdot),(l,g_2)}\right)^{-1}},\quad\quad\quad 1\leq k,l\leq G, \label{eq:informationmatrix}
%\end{equation}
\begin{equation}
	\left(H_{(i,j),\mathcal{G}}\right)_{k,l} = \text{Tr}\left\{\left(\nabla^2_{(i,j),(\mathcal{G},\mathcal{G})}\ell(B^\star)\right)
		\Lambda_{(k,l),(\cdot,\cdot)}\right\},\quad 1\leq k,l\leq G,, \label{eq:informationmatrix}
\end{equation}
where $\Lambda\in\real^{G\times G}\otimes\real^{G\times G}$ is a tensor such that $\text{Mat}(\Lambda)$ is a pseudoinverse of $\text{Mat}\left(\nabla^2_{(\mathcal{G},\mathcal{G}),(\mathcal{G},\mathcal{G})}\ell(B^\star)\right)$, and Mat is the operation that unfolds the entries of a tensor $\Lambda$ into a $G^2\times G^2$ matrix. The matrix $H_{(i,j),\mathcal{G}}$ measures how well the variable corresponding to the edge $(i,j)$ can be represented by the variables in the active subgraph.
\newtheorem{assump}{Assumption}
\begin{assump}[Restricted Strong Convexity] \label{assump:RSC}
	There exists a set $C\subset\real^{N\times N}$ with $B^\star\in C$, and constants $m>0$, $\tilde L<\infty$ such that
	\begin{align*}
		\sum_{i,j} \Delta_{i,j}\text{Tr}\left\{\left(\nabla^2_{(i,j),(\cdot,\cdot)}\ell(B)\right)\Delta\right\} %\Tr{\Delta^T\nabla^2\ell(B)\Delta} 
		& \geq  m \|\Delta\|_2^2, \quad \forall B\in C\cap \mathcal{M},\Delta\in C\cap \mathcal{M}\\
		\|\nabla^2\ell(B) - \nabla^2\ell(B^\star)\|_2 & \leq  \tilde L\|B-B^\star\|_2,\quad\forall B\in C.
	\end{align*}
\end{assump}
\begin{assump}[Irrepresentability] There exists a constant $0<\tau<1$ such that
	\begin{equation*}
		\max_{i\in\mathcal{G}^C}\left\|\left(\sum_{k=1}^G\|(H_{(i,j),\mathcal{G}})_{k\cdot}\|_2\right)_{j=1}^N\right\|_2= 1-\tau<0.
	\end{equation*}
	%  This is an alternative definition from Bach's paper	
	%	\begin{equation*}
	%	\max_{i\in\mathcal{G}^C}\frac{1}{2} \left\| \left(\Tr{H_{(i,j),\mathcal{G}} DB_{\mathcal{G}\mathcal{G}}^\star}\right)_{j=1}^N \right\|_2 = 1-\tau<0,
	%	\end{equation*}
	%	with $D$ a diagonal matrix such that $D_{ii}=1/\|B^\star_{g_i\cdot}\|_2$ for $g_1,\ldots,g_G\in \mathcal{G}$.
	\label{assump:irrepresentability}
\end{assump}

This version of the irrepresentability condition corresponds to the one usually employed in group lasso penalties \citep{bach2008consistency}, but as we will see later, due to overlaps in the groups it further requires a lower bound on $\rho$ to work for model selection.

The first two assumptions are stated directly as a function of the loss for a fixed design case, but they can be substituted with bounds in probability for the case of random designs. In order to obtain rates of convergence, we do require a distributional assumption on the first derivative of the loss. This assumption can be substituted with a bound on $\max_i\|\nabla\ell(B^\star)_{(i)}\|_2$ (see Lemma \ref{lemma:probabilitybound} in the Appendix).
\begin{assump}[Sub-Gaussian score function] \label{assump:subgaussianity} Each pair in the sample $(A,Y)$ is independent and comes from a distribution such that the entries of the matrix $\nabla\tilde\ell(Y,A;B^\star)$ are subgaussian.  That is, for all $t>0$ there is a constant $\sigma^2>0$ such that
	\begin{equation*}
		\max_{i,j} \prob\left(\|\nabla_{ij}\tilde\ell\left(Y, A; B^\star\right)\|_\infty>t\right) \leq 2\exp\left(-{t^2/\sigma^2}\right) \ . 
	\end{equation*}
\end{assump}
With these assumptions, we establish consistency and correct model selection. The proof is given in the Appendix..
\newtheorem{prop}{Proposition} 
\begin{prop}\label{prop:error}
	Suppose Assumptions \ref{assump:RSC} and \ref{assump:subgaussianity} hold. 
	\begin{enumerate}
		\item[(a)\label{prop-parta}] Setting the penalty parameters as  $\lambda = c_1\sigma\sqrt{\frac{\log N}{n}}\min\left\{\frac{\sqrt{N}}{1+\rho}, \frac{1}{\rho}\right\}$ and $\rho\geq 0$ for some constant $c_1>0$, with probability at least $1-2/N$ the optimal solution of \eqref{eq:main_problem} is unique and satisfies
		\begin{equation}
			\|\hat{B}-B^\star\|_2 = O_P\left(\sigma N\sqrt{\frac{\log N}{n}}\right). \label{eq:frobenius_rate_noIrrep}
		\end{equation}
		\item[(b)] Suppose Assumption \ref{assump:irrepresentability} also holds.  If  $n> c_2G^2\sigma^2\log N$ for a constant $c_2>0$, setting the penalty parameters as  $\lambda = c_3\sigma\sqrt{\frac{\log N}{n}}\min\left\{\frac{\sqrt{N}}{1+\rho}, \frac{1}{\rho}\right\}$ for some constant $c_3>0$, and 
		\begin{equation}
			\rho > \frac{1}{\tau} - \frac{1}{\sqrt{G}} \ , 	\label{eq:rho_lowerbound}
		\end{equation}
		then 
		\begin{align}
			\|\hat{B}-B^\star\|_2 & = O_P\left(\sigma G\sqrt{\frac{\log N}{n}}\right), \label{eq:frobenius_rate_Irrep}\\
			\Bbb{P}\left(\hat{\mathcal{G}}\subseteq \mathcal{G}\right) & = 1-2/N.\label{eq:support_consistency}
		\end{align}
	\end{enumerate}
\end{prop}

The part of the penalty associated with $\rho$ causes the solution to be sparse. Due to the overlap in the groups, a small value of $\rho$ will usually not result in zeros in the solution of the problem \eqref{eq:main_problem}.   The lower bound on $\rho$ in \eqref{eq:rho_lowerbound} ensures that the irrepresentability condition of \cite{lee2015model} holds (see Lemma \ref{lemma:irrepresent} in the Appendix). 

Proposition \ref{prop:error} ensures that, with high probability, all edges estimated to have non-zero weights are contained in the active subgraph, and quantifies the error in estimating the entries of $B^\ast$. To ensure that all active nodes are recovered, at least one edge corresponding to each active node needs to have a non-zero weight. A similar result can be obtained to guarantee recovery of all active nodes under a stronger assumption that the magnitude of the non-zero entries of $B^\star$ is bounded below by $|B_{ij}^\star|>c_4 G^2\lambda$ for a constant $c_4$. The condition in part (b) requires a sample size that grows faster than the size of the active subgraph, which in practice can be much smaller than the size of the network, making the method suitable for our applications in which the sample size is limited and the number of nodes is large. 
%%%%%%%%%%%%%%%%%%%%%%%%%%%%%%%%%%%%%%%%%%%%%%%%%%%%%%%%%%%%%%%%%%%%%%%%%%%%%%%%%%%%%%%%%%%%%

\section{Numerical results on simulated networks\label{sec:simulations}}

In this section, we evaluate the performance of our method using synthetic networks. We are interested in assessing the ability of the method to correctly identify predictive edges, its classification accuracy, and comparisons to benchmarks.  
We compare the different methods' edge selection performance in simulations using area under the curve (AUC).   
Brain connectomic networks are characterized by organization of nodes into communities \citep{bullmore2009complex}, in which nodes within the same community tend to have stronger connections than nodes belonging to different communities. In order to generate synthetic networks that mimic this property, we  introduce community structure using the stochastic block model (SBM)  \citep{holland1983stochastic}.    Before generating edges, we assign each node a community label,  $C_i$, where $C_i\in[K]$ for each $i\in[N]$. The node assignments are the same for all networks in the population. Given the community labels,  network edges are generated independently from a distribution that only depends on the community labels of the nodes associated with each edge. Since fMRI networks are real-valued networks, we generate edge weights from a Gaussian distribution, rather than the standard Bernoulli distribution normally used with the SBM.   Specifically, we draw each $A_{ij}$ independently from  $N(\mu_{C_iC_j}, \sigma^2)$, with $\mu\in\Bbb{R}^{K\times K}$ defined by 
\begin{equation*}
	\mu_{kl} = \left\{\begin{array}{cl}
		0.3, & \text{if $k=l$},\\
		0.1 & \text{if $k \neq l$},
	\end{array}\right.
\end{equation*}
and $\sigma^2=0.18$. These values were chosen to approximately match the distribution of edge weights in our datasets (see Section \ref{sec:data}). We set the number of nodes $N = 300$, with $K=12$ communities of size $25$ each. We work with undirected networks, so the adjacency matrices are symmetric, with 44,850 distinct edges. Although our method is able to scale to larger networks, this moderately sized setting is already highly computationally demanding for many of the comparison benchmarks.

To set up two different class distributions, we select a set of active nodes $\mathcal{G}$ first, defined by the nodes corresponding to some communities selected at random.
Then, we alter a set of differentiating edges $\mathcal{E}$ selected at random from $\mathcal{G}\times\mathcal{G}$ with probability $p$. For each edge $(i,j)\in\mathcal{E}$, the distribution in class $Y=-1$ is  $N(\mu_{C_iC_j}, \sigma^2)$, while the distribution in  class $Y=1$ is changed to  $N(0.2, \sigma^2)$. Figure \ref{fig:sim_example} shows example expected adjacency matrices for each class. We then generate 50 networks from each class, resulting in a sample size of $n=100$. We vary  $G=|\mathcal{G}|$ by changing the number of communities selected, and the value of $p$, to study the effect of the number of active nodes and the density of differentiating edges inside a subgraph.  The number of communities selected is set to 1 ($|\mathcal{G}|=25$), 2 ($|\mathcal{G}|=50$) and all communities ($|\mathcal{G}|=300$); note that in  the last scenario all nodes are active and hence the node structure is not informative at all. 

\begin{figure}
	\centering
	\begin{subfigure}[t]{0.5\textwidth}
		\centering
		\includegraphics[height=0.7\textwidth]{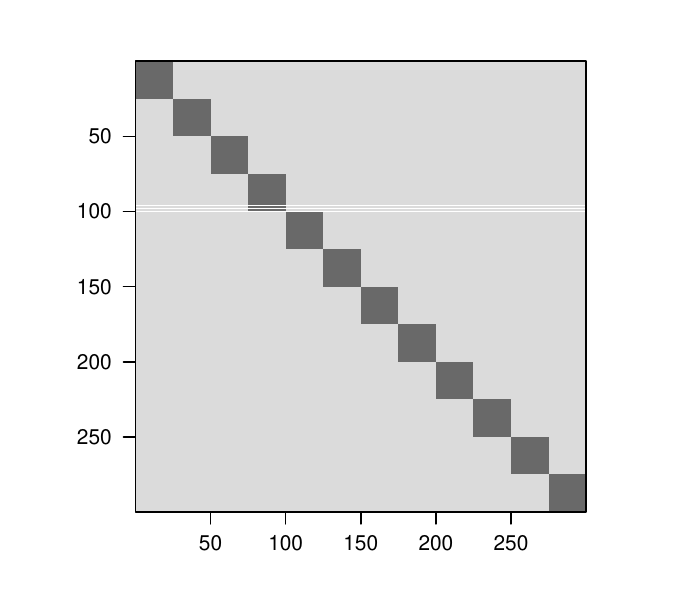}\\
		Class $Y=-1$
	\end{subfigure}%
	~ 
	\begin{subfigure}[t]{0.5\textwidth}
		\centering
		\includegraphics[height=0.7\textwidth]{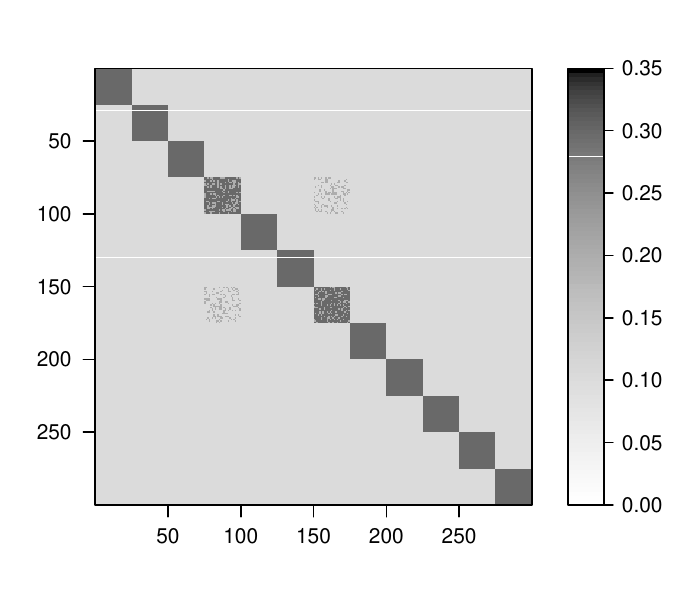}\\
		Class $Y=1$
	\end{subfigure}
	\caption{Expected adjacency matrices for each class. There are 50 active nodes $\mathcal{G}$ on communities 4 and 7, and edge weights on 25\% of the edges within $\mathcal{G}\times\mathcal{G}$ have been altered for the second class of networks ($Y=1$).\label{fig:sim_example}}
\end{figure}

Since we are interested is identifying predictive edges and nodes, we use the AUC of the receiver operating characteristic (ROC) curve, for both edge and node selection. For each method, we calculate the ROC curve by changing its corresponding sparsity parameter to vary the number of edges selected . For a selection method $\mathcal{M}$ and a sparsity parameter $\eta$ let $\hat{\mathcal{E}}(\mathcal{M}, \eta)$ be the set of edges selected by $\mathcal{M}$, and $\hat{\mathcal{G}}(\mathcal{M}, \eta)$ the set of active nodes corresponding to $\hat{\mathcal{E}}(\mathcal{M}, \eta)$. We calculate the edge false positive rate (EFPR) and edge true positive rate (ETPR) as
\begin{equation*}
	\text{EFPR}(\mathcal{M}, \eta)  =  \frac{\left|\hat{\mathcal{E}}(\mathcal{M}, \eta)\cap \mathcal{E}^C \right|}{\left|\mathcal{E}^C\right|}  \ , \ \ \ 
	\text{ETPR}(\mathcal{M}, \eta)  =  \frac{\left|\hat{\mathcal{E}}(\mathcal{M}, \eta)\cap \mathcal{E} \right|}{\left|\mathcal{E}\right|} \ .
\end{equation*}
The node FPR and TPR are calculated similarly. 

We also evaluate the prediction accuracy of the methods. For each method, we use 5-fold cross-validation to select the best tuning parameter using the training data, and then compute the test error on a different dataset simulated under the same settings.   The AUC and test errors reported are averaged over 30 replications.

Methods for benchmark comparisons on simulated networks were selected based on their good performance on real data  (see Section \ref{sec:data}). For our method (GC), we vary the parameter $\rho$ and compare results for two different values of $\lambda$, $.05$ (GC1) and $10^{-4}$ (GC2). For unstructured regularized logistic regression, we use the elastic net \citep{friedman2009glmnet}, with a fixed $\alpha=0.02$ (ENet). The performance of elastic net is not very sensitive to different values of $\alpha$, but the number of variables that the method is able to select with large values is limited (including the case of $\alpha=1$ that corresponds to the lasso). A support vector machine with $\ell_1$ penalty \citep{zhu20041,becker2009penalizedsvm} is also included (SVML1) for comparison, and additionally we evaluate the classification error of the original support vector machines (SVM)  \citep{cortes1995support}.  For both SVMs, we use linear kernels, which performed better than nonlinear ones. We also consider an independent screening method for variable selection based on the two sample $t$-statistic (T-stat). Finally, we also compare with the signal-subgraph method (SS) \citep{vogelstein2013graph} which is the only other method that takes into account the network structure of the predictor variables. Note that the signal subgraph is designed for binary networks, so in order to apply it we thresholded each edge at the population mean. For each method, we fit 10 different tuning parameters to change the sparsity of the solution.

\begin{figure}
	\centering
	\includegraphics[width=0.9\textwidth]{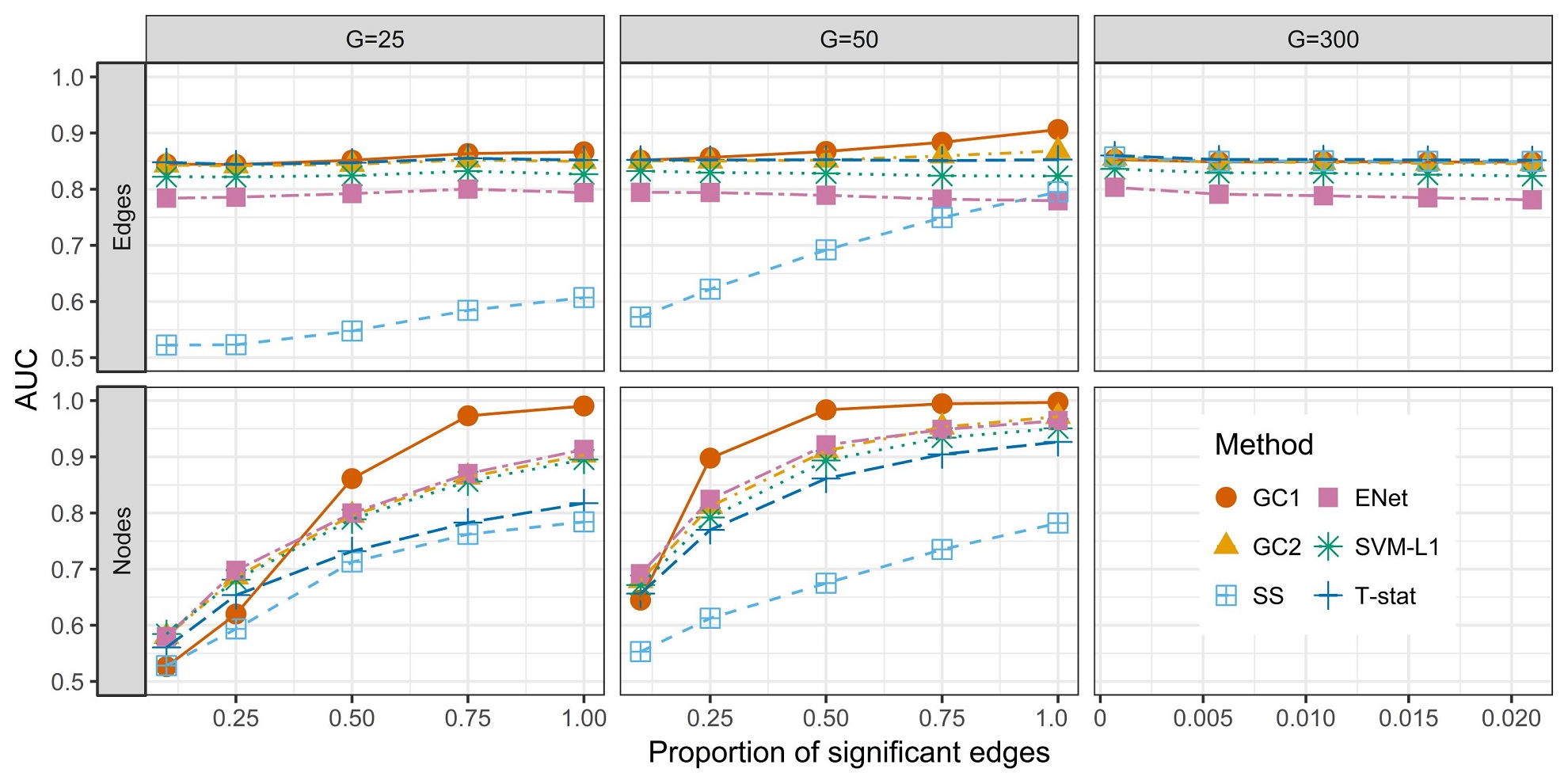}
	\caption{Variable selection performance of different methods in terms of edge AUC  (top) and node AUC (bottom) as a function of the fraction of differentiating edges in the subgraph induced by the active node set $\mathcal{G}$. As the proportion of active edges increases, methods that use network structure improve their performance when only a subset of the nodes is active. \label{fig:sims_auc}}
\end{figure}

\begin{figure}
	\centering
	\includegraphics[width=0.9\textwidth]{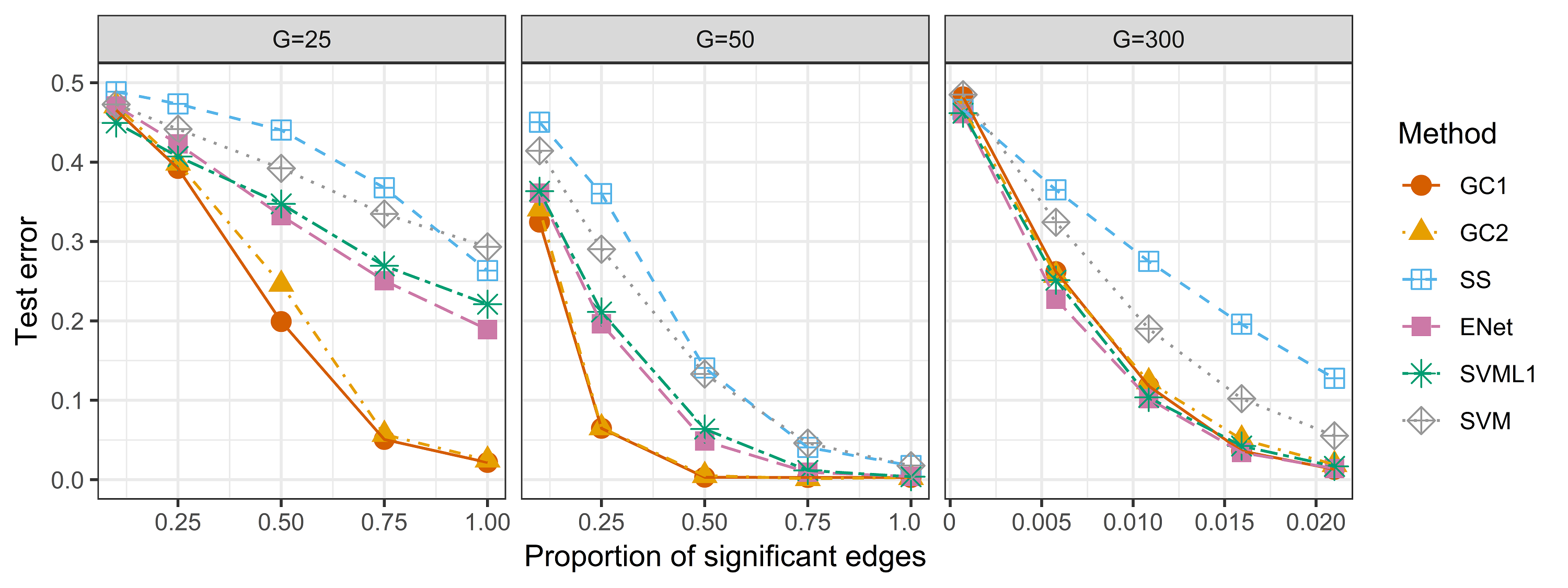}	
	\caption{Classification error of different methods as a function of the fraction of differentiating edges in the subgraph induced by the active node set $\mathcal{G}$. Our method is more accurate when only a subset of the nodes is active.  \label{fig:sims_cv}}
\end{figure}

Figure \ref{fig:sims_auc} shows the values of the average AUC for selecting edges (top) and nodes (bottom).   For $G=25$ and $50$, 
as the proportion of differentiating edges in the active subgraph increases, methods that take into account network structure (GC1, GC2 and SS) slightly improve their edge AUC, since enforcing node selection also results in better edge selection, while the edge AUC remains constant for unstructured methods (ENet, T-stat and SVML1).   On node selection, all methods improve the node AUC as the fraction of significant edges increases, but GC and SS have the largest gains. A similar trend is observed in classification error shown in Figure \ref{fig:sims_cv}. All methods improve as the proportion of differentiating edges increases, but our method has the best performance overall.  Our method performed the best with the larger value of $\lambda$ (GC1) on variable selection, particularly when the set of active nodes is smaller, but both values of $\lambda$ give very good classification performance.  In the last scenario ($G=300$), all nodes are active so the node AUC is undefined, and the node structure is not informative at all. Although the performance of our method is no longer the best, it performs comparably to state of the art methods that do not use network structure.

 In terms of computing time, since there are many contributing factors including the software choice for implementation and the tuning parameters, a fair comparison is difficult.    We can roughly say that elastic net is the fastest, taking about a minute to run a cross-validation instance, while our method takes about 10 minutes on average, and the  signal-subgraph takes more than an hour.
%%%%%%%%%%%%%%%%%%%%%%%%%%%%%%%%%%%%%%%%%%%%%%%%%%%%%%%%%%%%%%%%%%%%%%%%%%%%%%%%%%%%%%%%%%%%%%%%%%%%%%%%%%%%%%%%
\section{Application to schizophrenia data\label{sec:data}}
We analyze the performance of the classifier on the two different brain fMRI datasets previously described in Section \ref{sec:intro}. %, each containing schi\-zo\-phre\-nic patients and controls. The first dataset comes from the Center for Biomedical Research Excellence (COBRE). The second dataset, which we refer to as UMich data, is from the lab of Professor Stephan F. Taylor in the Department of Psychiatry at the University of Michigan. After performing pre-processing and registration steps on the raw images (see Appendix \ref{appendix:data} for a detailed description of these steps), we obtained a set of  
The code of our classifier and the processed connectomics datasets are available at \url{https://github.com/jesusdaniel/graphclass}.

\subsection{Classification results}

\begin{figure}
	\centering
	\begin{subfigure}[b]{0.32\textwidth}
		\includegraphics[width=\textwidth]{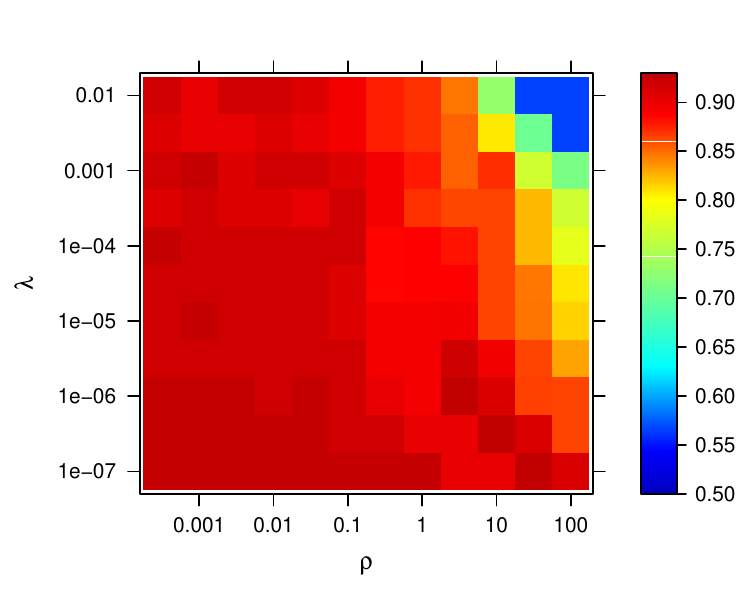}
		\caption{COBRE accuracy}
		\label{subfig:acc-cobre}
	\end{subfigure}
	\begin{subfigure}[b]{0.32\textwidth}
		\includegraphics[width=\textwidth]{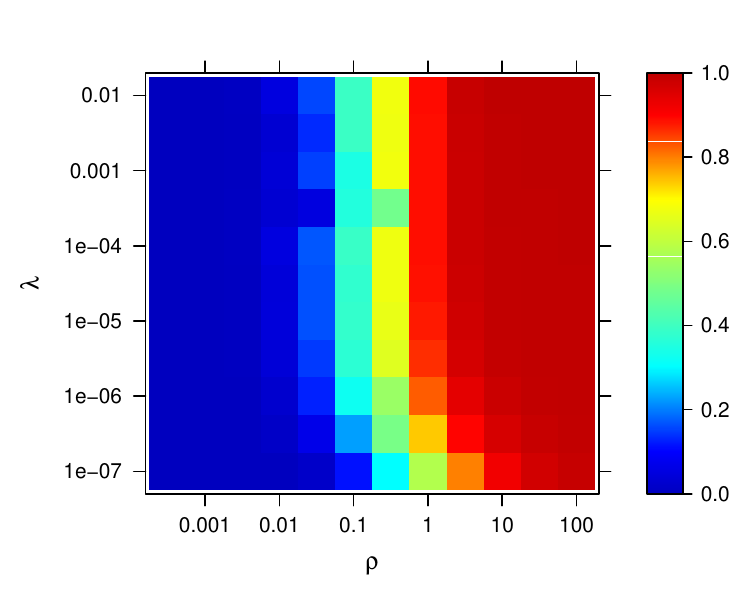}
		\caption{COBRE edge sparsity}
		\label{subfig:sp-cobre}
	\end{subfigure}
	\begin{subfigure}[b]{0.32\textwidth}
		\includegraphics[width=\textwidth]{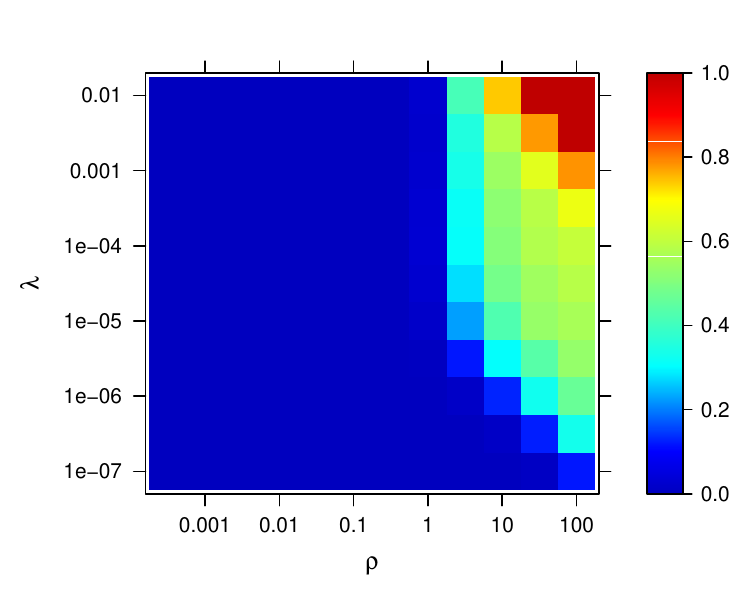}
		\caption{COBRE node sparsity}
		\label{subfig:ns-cobre}
	\end{subfigure}\\
	\begin{subfigure}[b]{0.32\textwidth}
		\includegraphics[width=\textwidth]{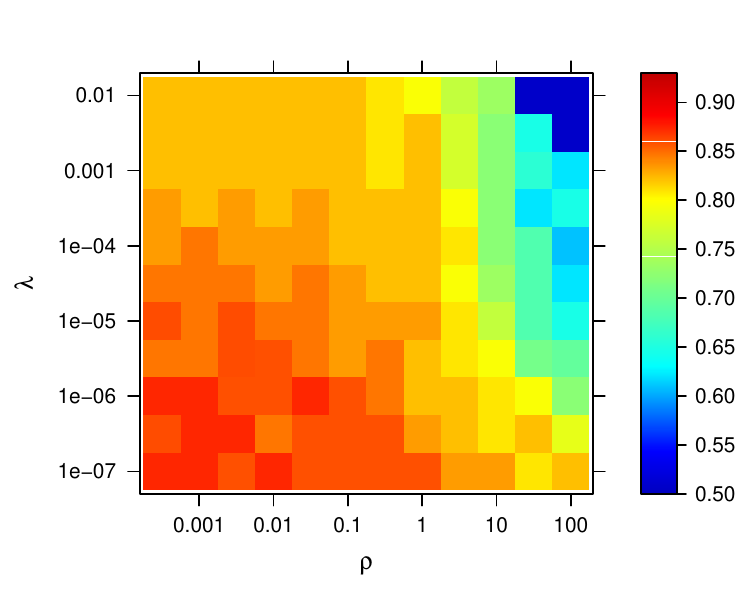}
		\caption{UMich accuracy}
		\label{subfig:acc-umich}
	\end{subfigure}
	\begin{subfigure}[b]{0.32\textwidth}
		\includegraphics[width=\textwidth]{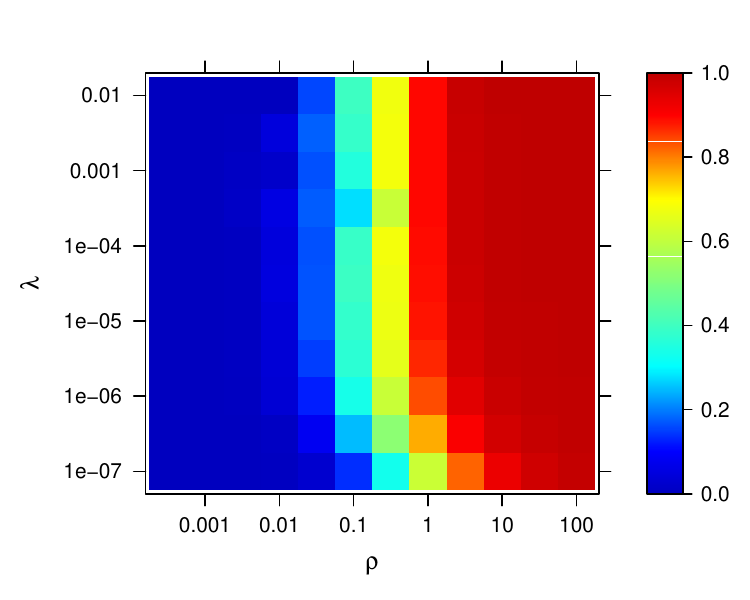}
		\caption{UMich edge sparsity}
		\label{subfig:sp-umich}
	\end{subfigure}
	\begin{subfigure}[b]{0.32\textwidth}
		\includegraphics[width=\textwidth]{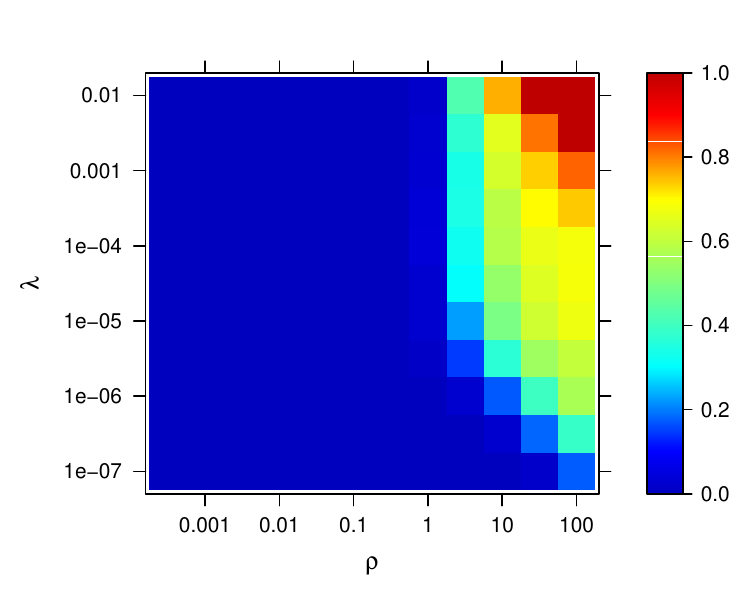}
		\caption{UMich node sparsity}
		\label{subfig:ns-umich}
	\end{subfigure}
	\caption{Cross-validated results for the two data sets.   Classification accuracy (left), fraction of zero edge coefficients (middle), and fraction of inactive nodes (right).\label{fig:cobre_average}}
\end{figure}

First, we evaluate our method's classification accuracy. We use a nested 10-fold cross-validation to choose tuning parameters and estimate the test accuracy. The classifier is trained for a range of values of $\lambda$ and $\rho$, with $\lambda\in \left\{10^{-7}, 10^{-6.5}, \ldots,  10^{-2}\right\}$ and $\rho\in \left\{10^{-3}, 10^{-2.5}, \ldots,  10^{2}\right\}$. The value of $\gamma$ in \eqref{eq:main_problem} is set to $10^{-5}$; we observed that setting $\gamma$ to a small value  speeds up convergence without affecting the accuracy or sparsity of the solution.   Figure \ref{fig:cobre_average} shows the average cross-validated accuracy, sparsity (fraction of zero coefficients) and node sparsity (fraction of inactive nodes), as a heat map over the grid of tuning parameter values.  We observe that  $\lambda$ has little influence on sparsity, which is primarily controlled by $\rho$. Moreover, as Proposition \ref{prop:error} suggests, values of $\rho < 1$ do not result in node selection.  As expected, accuracy generally decreases as the solution becomes sparser,  which is not uncommon in high-dimensional settings \citep{hastie2015statistical}. However, we can still achieve excellent accuracy with a substantially reduced set of features. In the COBRE dataset, the best accuracy is obtained with only 1886 edges (5.4\%) but almost all nodes are active (260).   On the UMich data, 29733 edges (85.6\%) achieve the best performance, and all nodes are active. Choosing parameters by cross-validation often tends to include too many noise variables \citep{meinshausen2007relaxed}, as we also observed in simulations. A commonly used technique to report solutions that still achieve good accuracy with a substantially reduced set of features is the so-called ``one-standard-error rule" \citep{hastie2015statistical}, in which one selects the most parsimonious classifier with cross-validation accuracy at most one standard error away from the best cross-validation accuracy.  Figure \ref{fig:fitted-solutions} shows the solutions for each dataset obtained by this rule. Nodes are ordered by brain systems (see Figure \ref{fig:powerparcellation}). The fitted solution for COBRE has  549 non-zero coefficients (1.56\%) and 217 active nodes (82.5\%), while the UMich solution has 11748 non-zero entries (33.8\%), and all nodes are active. Note that when many variables are selected, the magnitude of the coefficients becomes small due to the grouping effect of the penalty \citep{zou2005regularization}.

\begin{figure}
	\centering
	\begin{subfigure}[b]{0.48\textwidth}
		\includegraphics[width=\textwidth]{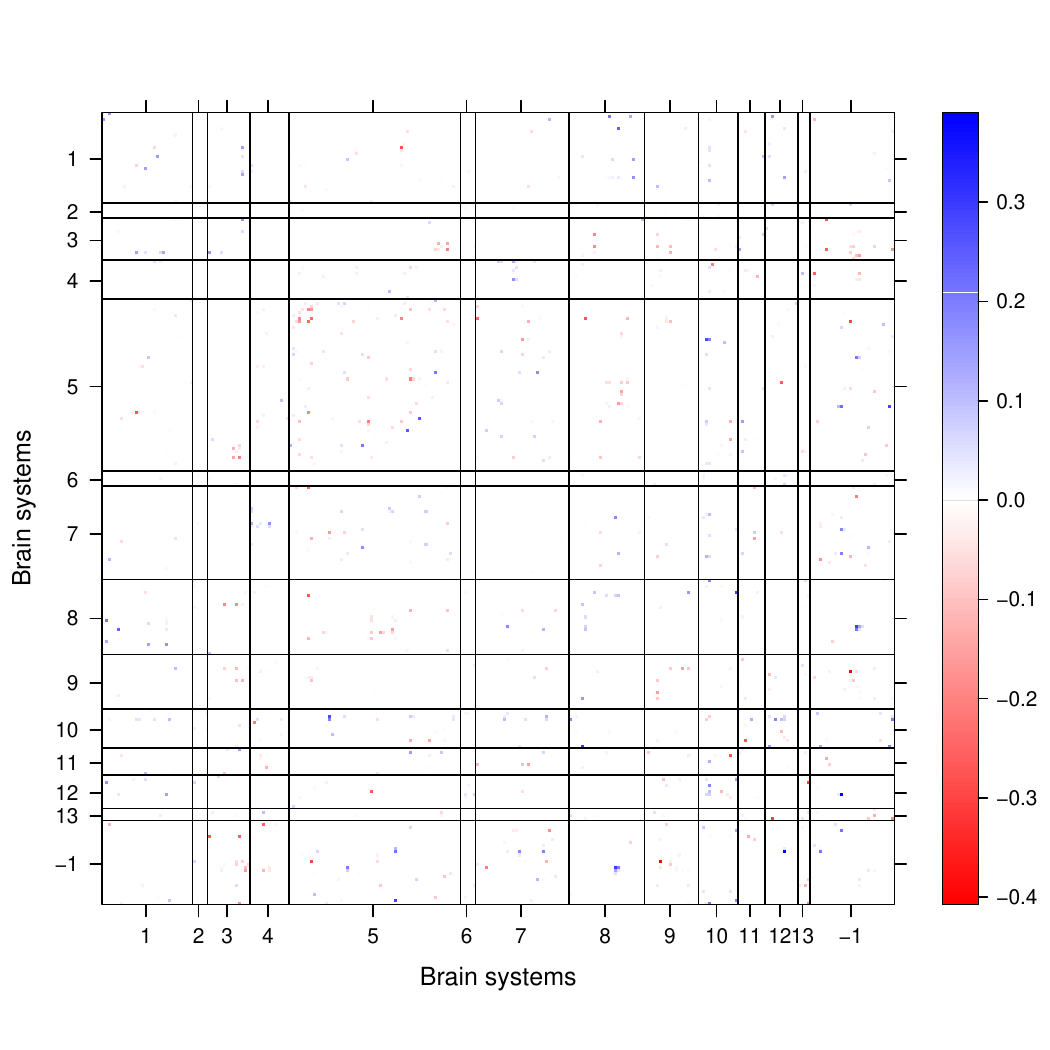}
		\caption{COBRE solution}
		\label{subfig:fitted-cobre}
	\end{subfigure}
	\begin{subfigure}[b]{0.48\textwidth}
		\includegraphics[width=\textwidth]{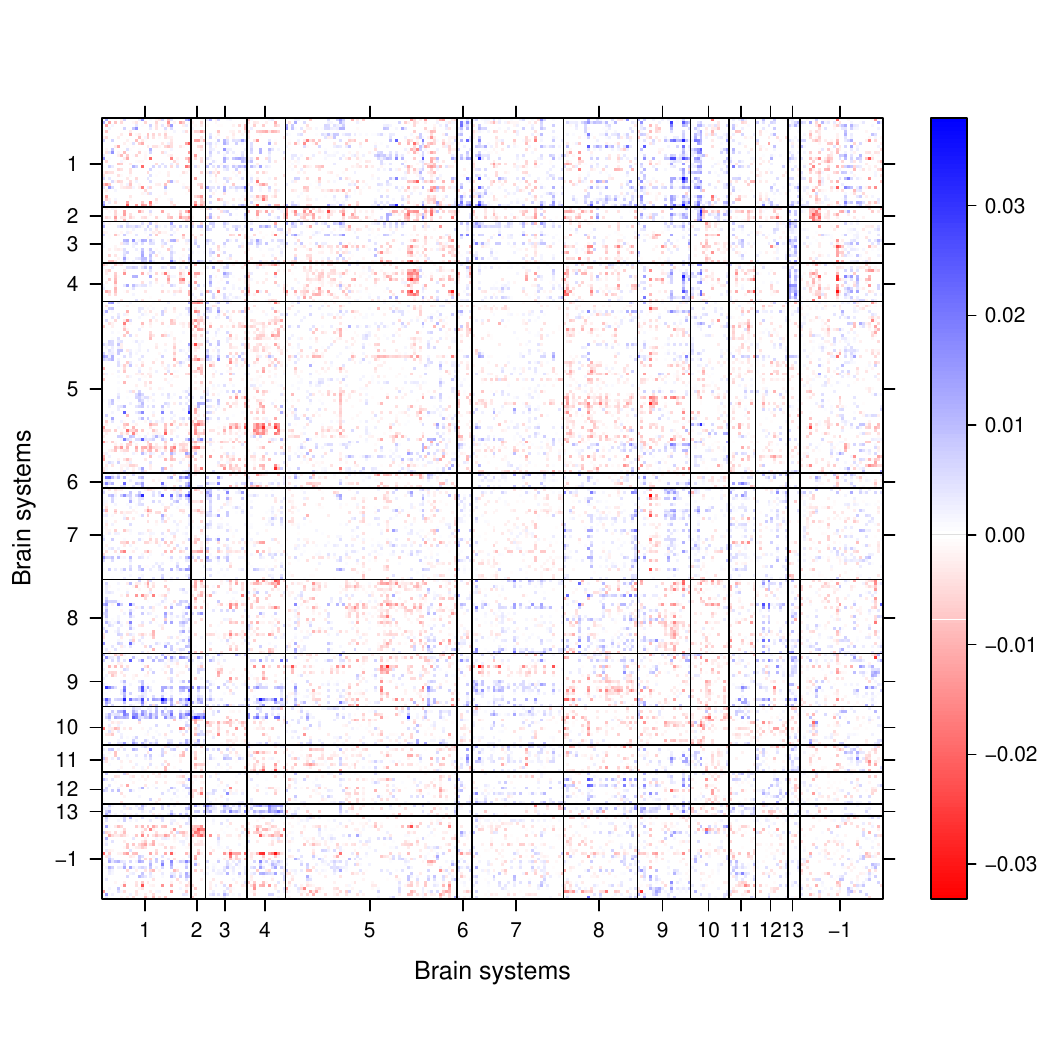}
		\caption{UMich solution}
		\label{subfig:fitted-umich}
	\end{subfigure}
	\caption{Fitted coefficients for COBRE and UMich datasets, with tuning parameters selected by the ``one standard error rule". Positive coefficients corresponds to higher edge weights for schizophrenic patients.%COBRE has 1092 \liza{Why is this number double the one in the text?  Check but no need to put this in the caption in any case.} non-zero edges and 217 active nodes, while UMich has 11748 non-zero edges and all nodes are active.
		\label{fig:fitted-solutions}}
\end{figure}

We also compared our method to benchmarks (Table \ref{tab:full_cv}), using the same methods as in the previous section and training and evaluating all methods using with the same nested 10-fold cross-validation. For SVM, we tested different kernels, including graph aware kernels \citep{gartner2003graph}, but in most cases local kernel methods were no better than random guessing.  We additionally included random forests  and a method based on global and local network summaries previously proposed as features for classifying brain data \citep{prasad2015brain}. For the latter, because our dataset is much larger, we only considered  global and node features proposed in \cite{prasad2015brain}, which resulted in about 30,000 features per individual, and omitted edge features.     \cite{watanabe2014disease} evaluated their classifiers on a different parcellation of the COBRE data, and we do not include their methods since they are based on the assumption of equally spaced nodes and cannot be directly applied to our data. Their reported accuracy of 71.9\% and 73.5\% for the COBRE data is substantially lower than our method, although the results are not directly comparable.

% we use a version of the method proposed in , in which several local and global network measures are computed and classification is done using a SVM. To compute these measures, we additionally constructed binary networks with a varying fraction of edges \cite{liu2008disrupted}, varying this fraction from 0.1 to 1 as in \cite{prasad2015brain}. Due to the much larger size of our networks, we restrict to real-valued and vector-valued network summaries,  resulting in a vector of size 29,120 for each individual. 

Results in Table \ref{tab:full_cv} show that most methods performed better on the COBRE dataset than on the UMich dataset, which can be partially explained by the different sample sizes and possibly noise levels.  Besides differences in sample size and demographic characteristics (Table \ref{table:patients}), the COBRE dataset is more homogeneous as it was collected using identical acquisition parameters, whereas the UMich dataset was pooled across five different experiments spanning seven years. 

Our method performs very well on both datasets, particularly among methods that  do variable selection. SVMs, which use the hinge loss, perform well too, and generally outperform methods using the logistic loss.   Our penalty can be combined with any loss, so we could also include our penalty combined with hinge loss which might potentially improve classification accuracy, but we do not pursue this direction, for two reasons:  one, our method is close to SVM + L1 as it is (better on COBRE, slightly worse on UMich but the difference is within noise levels), and because solutions based on logistic loss are generally considered more stable and preferable for variable selection \citep{hastie2015statistical}.   In Figure \ref{fig:dataVariableSelection}, we plot cross-validated classification accuracy of these methods as a function of the number of variables selected.  For the COBRE data, as we have observed before, good accuracy can be achieved with a fairly small number of edges, and the noisier UMich data requires more edges.   In all cases, our method uses fewer nodes than the others, as it is designed to do so.

\begin{table}
    {\footnotesize
	\centering
	\begin{tabular}{l c c}
		\hline & \multicolumn{2}{c}{\textbf{Classification accuracy \% (s.e.)}} \\ 
		\textbf{Method} &\textbf{COBRE }& \textbf{UMich } \\ 
		\hline 
		\textit{With variable selection} & & \\
		\textbf{Our method (GC)} & 92.7 (2.6) & 85.9 (3.6) \\ 
		Elastic net  & 89.5 (1.8) & 82.6 (4.7) \\
		SVM-L1 & 87.9 (2.2) & 86.2 (4.3)\\
		Signal-subgraph & 86.1 (3.3) & 82.4 (3.3)\\
		DLDA & 84.6 (3.3) & 73.4 (3.9)\\
		Lasso & 80.1 (5.6)  & 60.9 (5.6)  \\ 
		\textit{No variable selection} & & \\
		SVM & 93.5 (2.1) &  89.8 (2.5)\\
		Ridge penalty  & 91 (2.6) & 80.9 (3.5) \\
		Random forest & 74.2 (2.6) & 82.1 (3.9)\\
		Network summaries & 61.4 (3.1) & 65 (7.2)\\
		%Network summaries (svm radial)& 64.5 (3.8)& \\	
		\hline 
	\end{tabular} }
	\caption{Cross-validated accuracy (average and standard errors over 10 folds) for different methods.	\label{tab:full_cv} }
\end{table}

\begin{table}[ht]
	\centering
	{\footnotesize
	\begin{tabular}{l|rrr|rrr}
		\hline
		& \multicolumn{3}{c}{{COBRE}} &\multicolumn{3}{|c}{{UMich}}\\
		& {Edge} & {Systems} & {Coefficient} & {Edge} & {Systems}  & {Coefficient}\\ 
		\hline
		1 & (208, 85) &   (9, -1) & -0.187 & (110, 207) &   (5, 9) & -0.013\\
		2 & (260,  11) &  (12,  -1) & 0.183 & (255, 113) &   (1, 5) &  0.014\\ 
		3 & (194, 140) &   (8, -1) & 0.136 & (33, 218) &   (1, 9) & 0.016\\ 
		4 & (52, 186) &   (3, 8) & - 0.1 & (46, 225) &   (2, 10) &  0.013\\  
		5 & (160, 239) &  (7, 11) & -0.082 & (43,  90) &   (2, 5) & -0.013\\ 
		6 & (120, 116) &   (5, 5) & 0.099 & (23, 225) &   (1, 10) &  0.012\\ 
		7 & (57, 129) &   (3, 5) & -0.128 &  (66, 118) &   (4,  5) & -0.013\\
		8 & (24, 114) &   (1, 5) & -0.148 & (26, 145) &   (1, 7) &  0.013\\
		9 & (81, 179) &   (5, 8)  & -0.129  & (186, 254) &   (8, -1) &  0.012\\ 
		10 & (193, 140) &   (8, -1) & 0.153  &  (15, 134) &   (1, 6) &  0.011 \\ 
		11 & (178, 234) &  (8, 10)  & 0.146 & (76, 207) &   (5, 9) &  -0.012\\ 
		12 & (18, 194) &   (1, 8)  &  0.116  & (65, 84) &   (4, -1) &  -0.012\\  
		13 & (215, 207) &   (9, 9) & -0.076  &   (26, 122) &   (1, 5) & 0.012 \\
		14 & (90, 224) &  (5, 10)  & 0.123 & (33, 145) &   (1, 7) &  0.012\\ 
		15 & (112, 253) &  (5, -1) &  0.136 & (36, 224) &   (1, 10)  & 0.011 \\ 
		\hline
	\end{tabular}}
	\caption{\label{table:data-stabselect} Edges with the top 15 largest selection probabilities from stability selection. The first column shows the pair of nodes making the edge,  the second column the brain systems the nodes belong to in the Power parcellation, and the third column the fitted coefficient of the edge}
	%\caption{Table, cobre SS  6.059917e+00, umich ss 9.701748e+00}
\end{table}

\begin{table}
	\centering
	{\footnotesize
	\begin{tabular}{l | c c}
		\hline & \multicolumn{2}{c}{{Test data}} \\ 
		{Training data} &{COBRE }& {UMich } \\ 
		\hline 
		{COBRE} & 92.7 (2.6) & 73.5 (3.4)\\
		{UMich} &  78.3 (3.0) &  85.9 (3.6)\\ 
		\hline 
	\end{tabular} }
	\caption{Classification accuracy (cross-validation average and standard error) of the classifier fitted on one dataset and evaluated on the other. The intercept (the mean) is fitted on the test data  and the accuracy is estimated using 10-fold cross-validation on the test data. \label{tab:cv_bothdatasets} }
\end{table}

\begin{figure}
	\centering
	\includegraphics[width=0.9\textwidth]{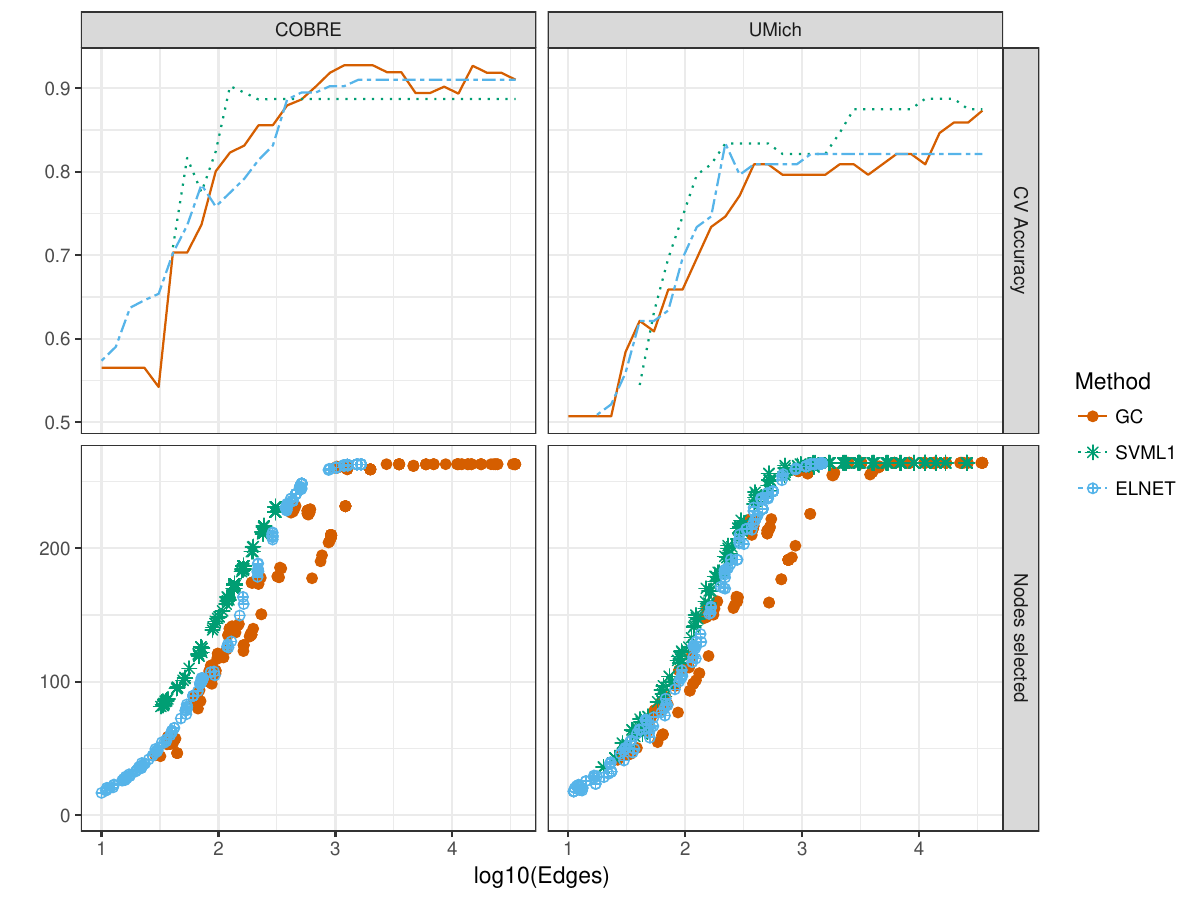}
	\caption{Cross-validated accuracy and number of nodes selected as a function of the number of edges used.} 
	\label{fig:dataVariableSelection}
\end{figure}

Ultimately, assessing significance of the selected variables is necessary, which is in general a difficult task in high-dimensional settings and an active area of research (see for example \cite{meinshausen2010stability,van2014asymptotically,lockhart2014significance,lee2016exact}). In brain connectomics, it is particularly challenging to identify significant variables because of small sample sizes \citep{button2013power}.   Here we employ \emph{stability selection} \citep{meinshausen2010stability} which can be shown to control a type of false discovery rate by employing many rounds of random data splitting and calculating the probability of each variable being selected.  Some versions of this method have been theoretically studied, and upper bounds on the expected number of variables with a low selection probability that are included in the final solution (i.e., errors) have been derived under mild conditions \citep{meinshausen2010stability,shah2013variable}.    We implemented the version of stability selection proposed by \cite{shah2013variable}, with values of $\lambda$ and $\rho$ obtained by cross-validation on the COBRE data, and by the ``one standard error rule" on the UMich dataset, since stability selection is most relevant to sparse solutions.   However, one of the advantages of stability selection is that it is not sensitive to the initial choice of tuning parameters, and changing tuning parameters  only slightly alters the ordering of variables with the largest selection probabilities.

The edges with the 15 largest selection probabilities are reported in Table \ref{table:data-stabselect}. 
Using the results of \cite{shah2013variable} (equation 8), we estimated that the expected number of falsely selected variables (variables with a probability of selection smaller than the estimated) is bounded by 6.1 for the COBRE dataset and 9.7 for the UMich data,  which also suggests that results on the UMich data might be less reliable. While the two datasets yield somewhat different patterns of edge selection, it is notable that the default mode network (5)  was often selected in both.   This network has been consistently implicated in schizophrenia \citep{whitfield2009hyperactivity,ongur2010default,peeters2015default}, as well as other psychiatric disorders, possibly as a general marker of psychopathology \citep{broyd2009default,menon2011large}. In the COBRE dataset, edges were also selected from the fronto-parietal task control region (8), previously
linked to schizophrenia \citep{bunney2000evidence,fornito2012schizophrenia}. These results coincide with the findings of \cite{watanabe2014disease} on a different parcellation of the same data, which is an encouraging indication of robustness to the exact choice of node locations. Some of the variables with the highest estimated selection probabilities appear in the uncertain system (-1), in particular in the cell connecting it with salience system (9), which suggests that alternative parcellations that better characterize these regions may offer a better account of the schizophrenia-related changes.  
Additionally, sensory/somatomotor hand region (1) and salience system (9) also stand out in the UMich data, and these are networks that have also been implicated in schizophrenia \citep{dong2017dysfunction}.

While results in Table \ref{table:data-stabselect} do not fully coincide on the two datasets, there are clear commonalities.     Table \ref{tab:cv_bothdatasets} compares classification accuracy when the classifier is trained on one dataset and tested on the other (with the exception of the intercept, since the datasets are not centered in the same way, which is fitted on a part of the test data, and the test error is then computed via 10-fold cross-validation).     While the accuracy is lower than when the same dataset is used for training and testing, as one would expect, it is still reasonably good and in fact better than some of the benchmark methods even when they train and test on the same data.  We again observe that the COBRE dataset is easier to classify.

Figure \ref{fig:selectednodes} shows the active nodes in the COBRE dataset (marked in green), corresponding to the endpoints of  the edges listed in Table \ref{table:data-stabselect}.   We also identified a set of 25 nodes that are not selected in any of the sparse solutions with cross-validation accuracy within one standard error from the best solution (marked in purple).   These consistently inactive nodes are mostly clustered in two anatomically coherent regions.

\begin{figure}
	\centering
	\includegraphics[width=0.25\textwidth]{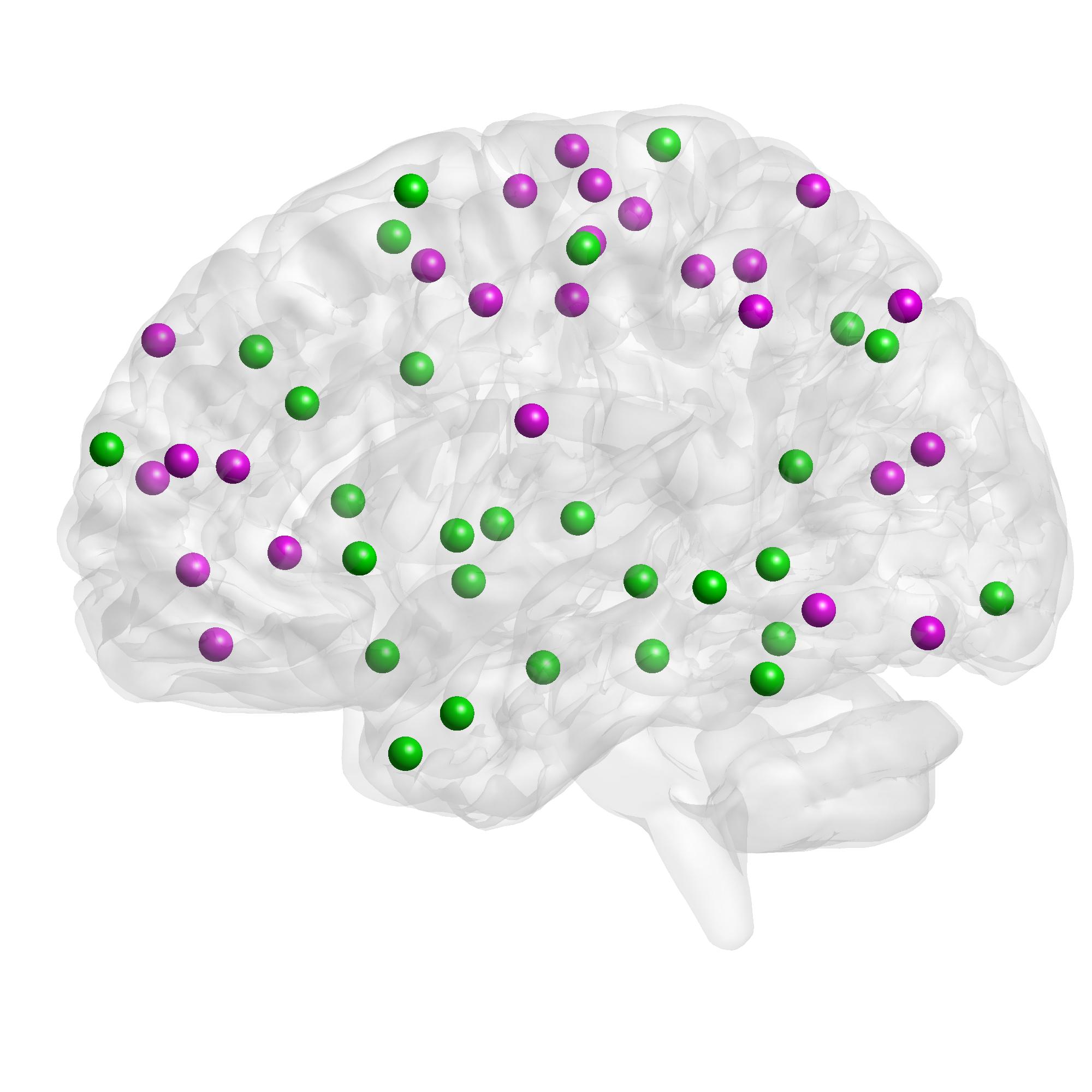}
	\includegraphics[width=0.25\textwidth]{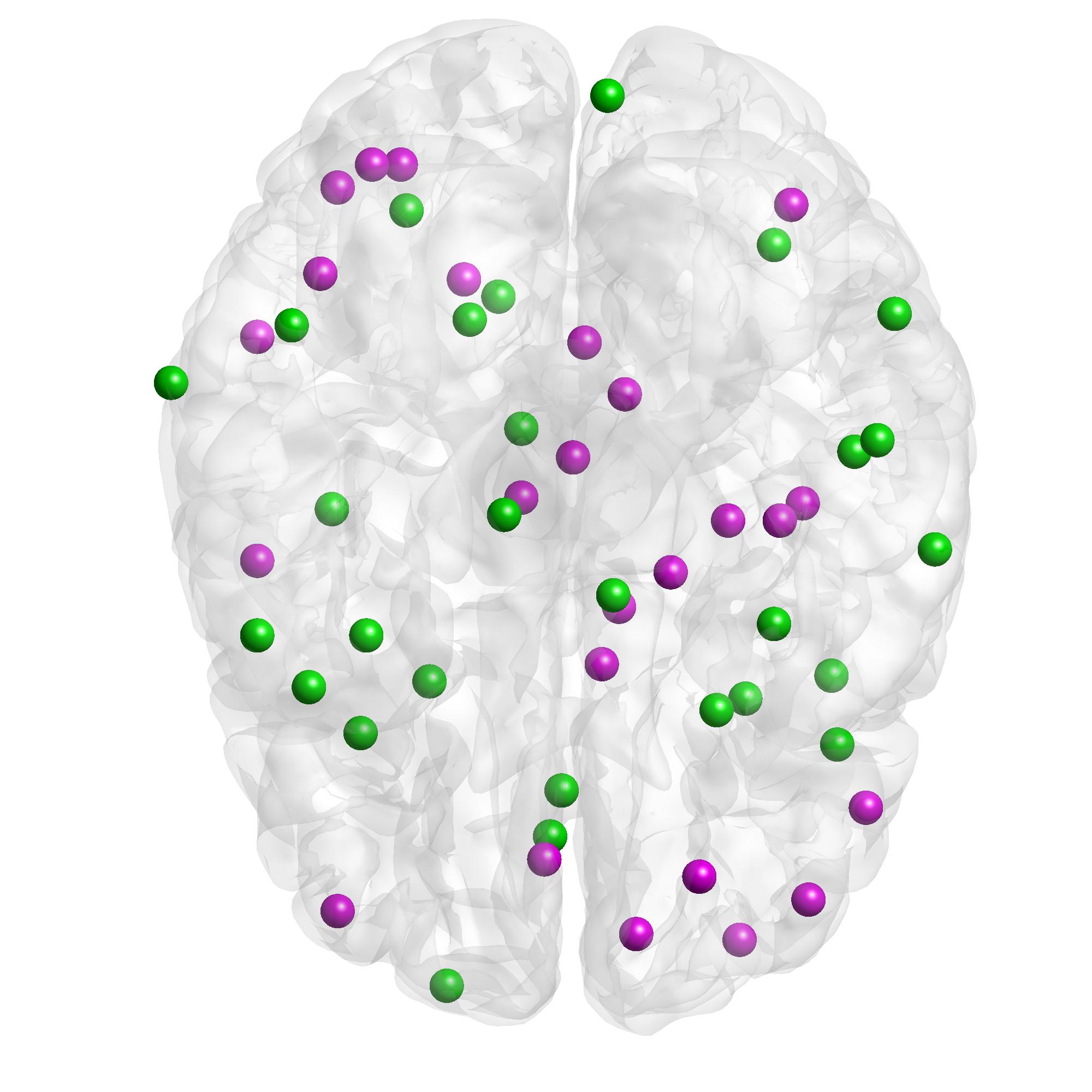}
	\includegraphics[width=0.25\textwidth]{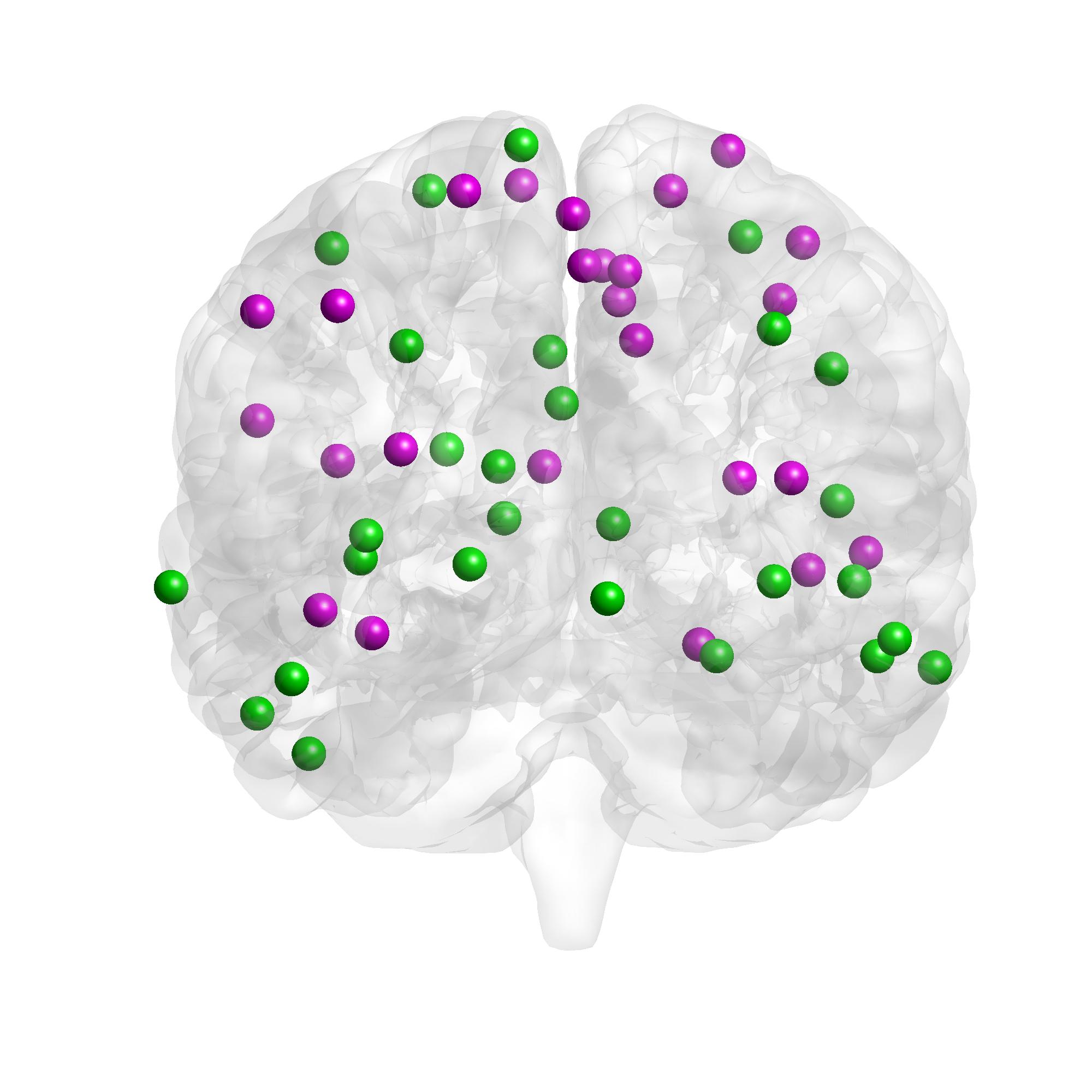}
	\caption{Nodes shown in green are endpoints of edges selected by stability selection shown in Table \ref{table:data-stabselect}.    Node shown in purple are nodes not selected by any of the sparse solutions within one standard error of the most accurate solution.  \label{fig:selectednodes}}
\end{figure}

\section{Discussion\label{sec:disc}}
We have presented a method for classifying graphs with labeled nodes, motivated by brain connectomics but generally applicable to any setting with such graphs.   The distinct feature of our method is that it is graph-aware, aiming to select a sparse set of {\em both} edges and nodes, but it is general in the sense that it does not rely on the spatial structure of the brain.   The method is computationally efficient since the regularization we use is convex, and the solution is implemented with efficient optimization algorithms. These properties guarantee fast  convergence to the solution, making the methods scalable to networks with thousands of nodes, which is enough to deal with many of the brain atlases usually employed in neuroimaging (see for example \cite{kiar2018high}). Statistically, the rate of convergence depends on the number of active nodes only, not the total number of nodes, which allows for accurate results with even moderate sample sizes if the active node set is small.

The results we obtained on the schizophrenia data are generally in agreement with previous studies.  In particular, the default mode network has been consistently implicated in schizophrenia and many other psychiatric disorders \citep{ongur2010default,broyd2009default}.     While different subnetworks were implicated by the two different datasets, we are still able to predict the disease status fairly accurately by training on one dataset and testing on the other.   The differences between the two datasets may reflect real differences in samples collected at different sites and in different experiments, as significant pathophysiological heterogeneity occurs for all psychiatric diagnoses, or they may simply reflect type 2 errors.

Our methods work very generally with a sample of networks with labeled nodes and associated responses.   The many pre-processing steps inevitable when dealing with fMRI data always  add some uncertainty,  and pre-processing decisions can potentially affect downstream conclusions.    We aimed to somewhat mitigate this  by using ranks, which are more robust and showed a slightly better performance on our datasets.  Another option, when practical, is to compare multiple pre-processing pipelines, and/or multiple measures of connectivity, to further validate results.   Our method's independence of these particular choices and its computational efficiency make it an attractive option for such comparisons.

\appendix

\section*{Acknowledgments}

This research was supported in part by NSF grant DMS-1521551, ONR grant N000141612910, and a Dana Foundation grant to E. Levina, NSF training grant DMS-1646108 support for D. Kessler,  as well as by computational resources and services provided by Advanced Research Computing at the University of Michigan, Ann Arbor.   S. F. Taylor's research  is supported by the National Institute of Mental Health (R01MH064148, R21MH086701, R21MH101676), the Boledovich Schizophrenia Research Fund and University of Michigan Clinical Translational Science Award (UL1RR024986).

\bibliographystyle{apalike}
\bibliography{Biblio}

\appendix

\section{Optimization algorithm details}
The optimization procedure for solving the penalized prediction problem introduced above consists in a proximal algorithm, and the steps are detailed in Algorithm \ref{alg:prox}. Each step requires to solve a further convex optimization problem via ADMM. The exact solution of the steps of this method is shown in Algorithm \ref{alg:admm}.

\label{appendix:optimization}
\begin{algorithm}
	\caption{Proximal algorithm for fitting graph classifier}
	\begin{algorithmic} 
		\Input Training sample $\{(A^{(1)},Y_1),\ldots,(A^{(n)},Y_n)\}$; regularization parameters $\lambda$, $\rho$; step size constants $\alpha\geq 1, \delta\in(0,1), \eta>0$; tolerance $\epsilon^{\text{\PROX}}>0$.
		\Initialize  Starting values  $B^{(0)}$, $t^{(0)}$.
		\Iterate for $k=1,2,\ldots$ until $\epsilon^{(k)} < \epsilon^{\text{\PROX}})$
		\begin{enumerate}
			\item Compute $W^{(k)}$ according to \eqref{eq:proximalalg_direction}.
			\item Compute $B^{(k)}$ by solving the proximal operator \eqref{eq:proximal_signalapprox}. \label{step:proxoperator}
			\item If condition \eqref{eq:backtrackingcondition} does not hold, decrease step size $t^{(k)} \leftarrow\delta t^{(k)}$ and return to \ref{step:proxoperator}.
			\item Calculate loss improvement   $\epsilon^{(k)} = \left\{\ell(B^{(k-1)}) + \Omega(B^{(k-1)})\right\}-\left\{\ell(B^{(k)}) + \Omega(B^{(k)})\right\}$ . 
			\item If $|\epsilon^{(k)}-\epsilon^{(k-1)}|/\epsilon^{(k)} < \eta$, increase step size $t^{(k+1)}=\alpha t^{(k)}$, otherwise set $t^{(k+1)} = t^{(k)}$.
		\end{enumerate} 	
		\Output $\hat{B}=B^{(k)}$.
	\end{algorithmic}
	\label{alg:prox}
\end{algorithm}

\begin{algorithm}
	\caption{Proximal operator by ADMM}
	\begin{algorithmic} 
		\Input $Z$, $\epsilon^{\ADMM}$, $\mu$.
		\Initialize $\tilde B^{(0)}=Z$, $R^{(0)}=Z$, $Q^{(0)}=Z$, $U^{(0)}=0_{N\times N}$, $V^{(0)}=0_{N\times N}$.
		\Iterate for $l=0,1,2,\ldots$ until convergence ($\epsilon^{(\counter)}_{\text{\ADMM-p}}<\epsilon^{\ADMM}$ and $\epsilon^{(\counter)}_{\text{\ADMM-d}}<\epsilon^{\ADMM}$)
		\begin{enumerate}
			\item Perform coordinate gradient descent on \eqref{eq:lagrangian} by computing 
			\begin{align*}
				\tilde{B}^{(l+1)} & =  \frac{1}{1+2\mu}\left\{Z + \frac{\mu}{2}\left(Q^{(l)}+Q^{(l)^T}\right) + \mu R^{(l)} - U^{(l)} -V^{(\counter-1)}\right\},\\
				Q^{(l+1)}_{(i)} & =  \left(1-\frac{t\lambda}{\mu\left\|\tilde{B}^{(l+1)}_{(i)} + \frac{1}{\mu}U^{(l)}_{(i)}\right\|_2}\right)_+\left(\tilde{B}^{(l+1)}_{(i)} + \frac{1}{\mu}U^{(l)}_{(i)}\right),\  i\in[N],\\
				R^{(l+1)}_{ij} & =  \left(1- \frac{t\lambda \rho}{\mu\left|\tilde{B}^{(l+1)}_{ij} + \frac{1}{\mu}V^{(l)}_{ij}\right| }\right)_+(\tilde{B}^{(l+1)}_{ij} + \frac{1}{\mu}V^{(l)}_{ij}), \  i,j\in[N],\\
				U^{(l+1)} & =  U^{(l)} + \mu\left\{\tilde{B}^{(l+1)}-\frac{1}{2}\left(Q^{(l+1)}+Q^{(l+1)^T}\right)\right\},\\
				V^{(l+1)} & = V^{(l)} + \mu\left(\tilde{B}^{(l+1)}-R^{(l+1)}\right).
			\end{align*}
			\item Update primal and dual residuals $\epsilon^{(\counter)}_{\text{\ADMM-p}}$ and $\epsilon^{(\counter)}_{\text{\ADMM-d}}$ 
			\begin{align*}
				\epsilon^{(l+1)}_{\text{\ADMM-p}} & =  \mu\left(\|Q^{(l+1)}-Q^{(l)}\|_\infty+\|R^{(l+1)}-R^{(l)}\|_\infty\right),\\
				\epsilon^{(l+1)}_{\text{\ADMM-d}} & =  \mu\left(\|\tilde{B}^{(l+1)}-Q^{(l+1)}\|_2+\|\tilde{B}^{(l+1)}-R^{(l+1)}\|_2\right).
			\end{align*}
		\end{enumerate}
		\Output $\tilde{B}=\tilde{B}^{(l+1)}$.
	\end{algorithmic}
	\label{alg:admm}
\end{algorithm}

\section{Proofs}
\label{appendix:theory}

Here we prove the bounds on Frobenius norm error and probability of support selection in Proposition \ref{prop:error}, following the framework of \cite{lee2015model} based on geometrical decomposability. A penalty $\Omega$ is geometrically decomposable if it can be written as \[\Omega(B)=h_A(B)+h_I(B) + h_{E^\perp}(B)\]
for all $B$, with $A,I$ closed convex sets, $E$ a subspace, and $h_C$ the support function on $C$ defined as $h_C\left(B\right)=\sup\left\{\left\langle Y,B\right\rangle  \left| Y\in C\right.\right\}$ .

The proof proceeds in the following steps.  
Lemma \ref{lemma:geom_decomp} shows that  an equivalent form of our penalty \eqref{eq:GL_penalty} is geometrically decomposable, allowing us to use the framework of \cite{lee2015model}.
Lemma \ref{lemma:probabilitybound}  shows the Assumption \ref{assump:irrepresentability} together with a lower bound on $\rho$ imply that the irrepresentability assumption of \cite{lee2015model} holds.  Assumption \ref{assump:irrepresentability} is directly on the entries of the loss Hessian, which simplifies the very general form of the assumption in \cite{lee2015model}.  Lemma \ref{lemma:probabilitybound} gives a bound on the entries of the loss gradient under the sub-gaussianity assumption \ref{assump:subgaussianity}.   Lemma \ref{lemma:compconsts} gives explicit bounds for the compatibility constants that appear on Theorem 1 of \cite{lee2015model}.   Finally, we combine these results to prove Proposition \ref{prop:error}.

Without loss of generality, to simplify notation we assume that $\mathcal{G}=\{1,\ldots,G\}$, that is, the active subgraph is in the first $G$ rows of the matrix. 
\begin{lemma}
	\label{lemma:geom_decomp}
	The penalty \eqref{eq:GL_penalty} can be written as geometrically decomposable.  
\end{lemma}

\begin{proof}[Proof of Lemma \ref{lemma:geom_decomp}]   We use an equivalent formulation of the penalty in which every coefficient is penalized only once. Let $B',B''\in\real^{N\times N}$ be matrices such that the upper triangular part of $B''$ and the diagonals of $B'$ and $B''$ are zero. Define
	\begin{equation*}
		\tilde\Omega(B',B'') = \sum_{i=1}^N\|B'_{(i)}\|_2 + \rho\|B''\|_1,
	\end{equation*}
	and $E = \{(B',B'')\in\real^{N\times 2N}\ :  \ B' = B'^{T}, B''_{ij}=B'_{ij}, \text{ for }i<j\text{ and }B''_{ij}=0\text{ for }i\geq j\}$.
	Denote by $R$ the transformation from $\real^{N\times N}$ to $\real^{N\times 2N}$ that replicates entries appropriately, 
	\begin{equation}
		(RB)_{ij} = \left\{ \begin{array}{cl}
			B_{ij}& \text{ if $1\leq j\leq N$}\\
			B_{i(j-N)}&\text{ if $j> N$}.\\
		\end{array}\right.
		\label{eq:repeating_operator}
	\end{equation}
	Therefore, for any $B\in\real^{N\times N}$, we can uniquely define $RB = (B',B'')$ such that $\Omega(B)=\tilde\Omega(B',B'')$. 
	We then show that $\tilde\Omega$ is geometrically decomposable. Moreover, for any $(B',B'')\in E$ we can define $R^{-1}$, so the penalties $\Omega$ and $\tilde\Omega$ on $E$ are equivalent. Define the sets $A,I\subset\real^{N\times 2N}$ such that
	\begin{align*}
		A =&  \left\{(B',B'')  :   \ \max_{i\in\mathcal{G}}\|B_{(i)}'\|_2 \leq 1, \max_{i\in\mathcal{G}^C}\|B_{(i)}'\|_2=0, \right. \\
		& \left. \ \max |B''_{ij}| \leq \rho, B_{ij}'' = 0,(i,j)\in\left(\mathcal{G}\times\mathcal{G}\right)^C \right\},\\
		I = &  \left\{(B',B'')  : \  \max_{i\in\mathcal{G}^C}\|B_{(i)}'\|_2 \leq 1, \max_{i\in\mathcal{G}}\|B_{(i)}'\|_2=0,\right.\\
		& \left. \  \max |B''_{ij}| \leq \rho,  B_{ij}'' = 0,(i,j)\in\mathcal{G}\times\mathcal{G} \right\} .
	\end{align*}		
	Letting $\left\langle Y,(B',B'')\right\rangle=\text{Tr}(Y'B'^T) + \text{Tr}(Y''B''^T)$, combining the arguments of \cite{lee2015model} for lasso and group lasso penalties, 
	\begin{align*}
		h_A(B',B'') =  & \sum_{i\in\mathcal{G}}\|B'_{(i)}\|_2 + \rho\sum_{(i,j)\in\mathcal{G}\times\mathcal{G}}|B''_{ij}|,\\
		h_I(B',B'') = & \sum_{i\in\mathcal{G}^C}\|B'_{(i)}\|_2 + \rho\sum_{(i,j)\in\left(\mathcal{G}\times\mathcal{G}\right)^C}|B''_{ij}|,\\
		h_E(B',B'') = & \left\{\begin{array}{cl}
			0 & \text{if }(B',B'')\in E,\\
			\infty & \text{otherwise}.
		\end{array}\right.
	\end{align*}
	Hence,  $\Omega$ can be written as a geometrically decomposable penalty 
	$$\Omega(B) = \tilde\Omega(B',B'') = \lambda\{h_A(B',B'')+h_I(B',B'')+h_E(B',B'')\}.$$
\end{proof}
We introduce some notation in order to state the irrepresentability condition of \cite{lee2015model}. For a set $F\subset\real^{N\times 2N}$ and $Y\in\real^{N\times 2N}$, denote by $\gamma_F\left(Y\right) = \inf\left\{\lambda>0 : Y\in F\right\}$ the gauge function on $C$. Thus,
\begin{equation*}
	\gamma_I(B',B'') = \max\left\{ \max_{i\in\mathcal{G}^C}\|B_{(i)}'\|_2,\ \frac{1}{\rho}\max_{(i,j)\in(\mathcal{G}\times\mathcal{G})^C}|B_{ij}''| \right\} + \textbf{1}_{I}(B',B''),% \leq \max\left\{1,\frac{1}{\rho}\right\}\max_{i\in\mathcal{G}^C}\|B_{(i)}^{(1)}\|_2kind of(under E),
\end{equation*}
where $\textbf{1}_{I}(B)= 0 $ if $B \in I$ and $\infty$ otherwise. Define 
\[V(Z)=\inf\{\gamma_I(Y) : \ \ Z-Y\in E^\perp,Y\in\real^{N\times 2N}\}\]
for $Z\in\real^{N\times 2N}$. Let $\tilde{\mathcal{M}}=E\cap\text{span}(I)^\perp$ be the set of matrices with correct support in the extended space $\real^{N\times 2N}$, similarly to $\mathcal{M}$ in \eqref{eq:m_activeset}.
Denote by $P_M$ and $P_{M^\perp}$ the projections onto $\tilde{\mathcal{M}}$ and $\tilde{\mathcal{M}}^\perp$. Define the function $\mathcal{H}(Z):\real^{N\times N}\rightarrow\real^{N\times N}$ as 
\begin{equation*}
	\mathcal{H}(Z)_{ij} = \left\{\begin{array}{cl}
		\Tr{H_{(i,j),\mathcal{G}} (P_M Z)_{\mathcal{G},\mathcal{G}}} & \text{if $j\in\mathcal{G}$},\\
		0 & \text{otherwise}.
	\end{array}\right.
\end{equation*}
where $H_{(i,j),\mathcal{G}}$ is the matrix defined in \eqref{eq:informationmatrix}. The Irrepresentability Assumption 3.2 of \cite{lee2015model} requires the existence of $0<\tilde\tau<1$ such that
\begin{equation}
	\sup_{Z\in A} V\left\{P_{M^\perp}\left(R\mathcal{H}(Z) - Z\right)\right\} < 1 - \tilde\tau. \label{eq:irrep_leeetal}
\end{equation}
For a support function $h$, denote by $\partial h(M)=\bigcup_{Y\in M}\partial h(Y)$ the set of subdifferentials of $h$ in $M$. Note that $\partial h_A(M) = A$, since $0\in M$ and $\partial h_A(0) = A$. 
%%%%%%%%%% Lemma %%%%%%%%%%%%%%%%%%%%%%%%%%%%%%%
\begin{lemma}
	If Assumption \ref{assump:irrepresentability} holds and $\rho>\frac{1}{\tau}-\frac{1}{\sqrt{G}}$, then there exists $0<\tilde\tau<1$ such that \eqref{eq:irrep_leeetal} holds.
	\label{lemma:irrepresent}
\end{lemma}

\begin{proof}[Proof of Lemma \ref{lemma:irrepresent}.]	Since $V$ is sublinear (Lemma 3.3 of \cite{lee2015model}),
	\begin{equation}
		\sup_{Z\in A} V\left\{P_{M^\perp}\left(R\mathcal{H}(Z) - Z\right)\right\}   \leq \sup_{Z\in A} V\left\{P_{M^\perp}\left(R\mathcal{H}(Z)\right)\right\} + \sup_{Z\in A} V\left\{P_{M^\perp}Z\right\}. \label{eq:irrep_VZ}
	\end{equation}
	To bound the first term, note that $E^\perp=\{(Z',Z'')|Z_{ij}'+Z_{ji}'+Z_{ij}''=0, j < i\}$.	Hence,
	\begin{align*}
		V(Y',Y'') = & \inf\left\{\gamma(U',U''):  U_{ij}'-Y_{ij}^{(1)}+ U_{ji}'-Y_{ji}' + U_{ij}''-Y_{ij}''=0,\  j < i \right\}\\
		 \leq & \inf\left\{\gamma(U',U''): \ \ U_{(i)}' = Y_{(i)}', i\in\mathcal{G}^C;  U_{\mathcal{G}^C,\mathcal{G}^C}'' = Y_{\mathcal{G}^C,\mathcal{G}^C}'';\right.\\
		 & \quad\quad\left. (U',U'')- (Y', Y'')\in E^\perp  \right\}\\
	     \leq & \max\left\{ \max_{i\in\mathcal{G}^C}\|Y_{(i)}'\|_2,\ \frac{1}{\rho} \left\|  Y_{\mathcal{G}^C,\mathcal{G}^C}''\right\|_\infty \right\}.
	\end{align*}
	Therefore,
	\begin{align*}
		V\left(P_{M^\perp}\left(R\mathcal{H}(Z)\right)\right)  \leq & \max\left\{ \max_{i\in\mathcal{G}^C}\|(P_{M^\perp}(R\mathcal{H}(Z)))_{(i)}^{(1)}\|_2,\ \frac{1}{\rho} \|(P_{M^\perp}(R\mathcal{H}(Z)))^{(2)}_{\mathcal{G}^C,\mathcal{G}^C}\|_\infty \right\}\\
		= &  \max_{i\in\mathcal{G}^C}\left\|\mathcal{H}(Z)_{(i)}\right\|_2,
	\end{align*}
	which implies that
	\begin{align}
		\sup_{Z\in A} V\left(P_{M^\perp}\left(R\mathcal{H}(Z)\right)\right)  \leq & \sup_{Z\in A} \left\{\max_{i\in\mathcal{G}^C}\left\|\mathcal{H}(Z)_{(i)}\right\|_2\right\}\nonumber\\
		\leq  & \sup_{B\in\real^{G\times G},\|B_{(i)}\|_2\leq 1}\left\{\max_{i\in\mathcal{G}^C}\left\|\left(\Tr{H_{(i,j),\mathcal{G}}B}\right)_{j=1}^N\right\|_2\right\}\nonumber\\
	 \leq & \max_{i\in\mathcal{G}^C}\left\|\left(\sum_{k=1}^G\|(H_{(i,j),\mathcal{G}})_{k\cdot}\|_2\right)_{j=1}^N\right\|_2  =  1-\tau. \label{eq:irrep_firsttermbound}
	\end{align}
	%\jesus{note that the last equation will be usually larger than lasso irrepresentability condition. to make it strictly smaller, it is necessary to divide by N, but result doesn't follow directly}
	%\jesus{Note also that it might be possible to improve this bound to let it depend just in subset $\mathcal{G}^C\times\mathcal{G}^C$ instead of the full set of inactive nodes}
	Let $Z=(Z',Z'')\in A$. Without loss of generality, assume that $Z'_{\mathcal{G},\mathcal{G}}=Z''_{\mathcal{G},\mathcal{G}}=0$ (note that these entries do not change $V(P_{M^\perp}Z)$). Therefore, $P_{M^\perp}Z=Z$. Hence, 
	\begin{align*}
		V(Z) = & \inf\left\{\gamma\left(U',U''\right): \ \ \left(U',U''\right)\in I,\ \left(U',U''\right)-\left(Z',Z''\right)\in E \right\}\\
		= & \inf\left\{\gamma\left(U',U''\right): \ U'_{ij} + U''_{ij}=Z'_{ji}, 1\leq j \leq G, G< i \leq N \right\}\\
		= & \inf\left\{\max\{\max_{i\in\mathcal{G}^C}\|U'_{(i)}\|_2,\frac{1}{\rho} \max_{\substack{j \in\mathcal{G}\\ i\in\mathcal{G}^C}}|U''_{ij}|\} : \ U'_{ij} + U''_{ij}=Z'_{ji}, j \in\mathcal{G}, i\in\mathcal{G}^C \right\}\\
		\leq & \inf\left\{\max\{\max_{i\in\mathcal{G}^C}\|U'_{(i)}\|_2,\frac{1}{\rho} \max_{\substack{j \in\mathcal{G} \\ i\in\mathcal{G}^C}}|U''_{ij}|\} : \ U'_{ij} + U''_{ij}=1, j \in\mathcal{G}, i\in\mathcal{G}^C\right\}
	\end{align*}
	The last bound from  $|Z'_{ji}|\leq 1$ and no longer depends on $Z$.
	It is easy to see that the minimum is attained when, for each $i>G$, 
	\[\left\|U'_{(i)}\right\|_2 = \frac{1}{\rho} \left|U''_{ij}\right|, \  \quad 1\leq j\leq G, \]
	and therefore 
	\begin{equation}
		V(Z) \leq \frac{\sqrt{G}}{1+\rho\sqrt{G}}.\label{eq:irrep_rhobound}
	\end{equation}	
	Moreover, if $Z^\ast\in A$ is defined such that $(Z^\ast)^{(1)}_{G+1,i}=1$ for $i=1,\ldots,G$ and $0$ elsewhere, then $V(Z^\ast)$ achieves this bound, which shows that $\rho > 1- \frac{1}{\sqrt{G}}$ is a necessary condition for the irrepresentability to hold, even in the case where the entries of the Hessian that denote the information between active and inactive edges is zero. Therefore, plugging the bounds \eqref{eq:irrep_firsttermbound} and \eqref{eq:irrep_rhobound} into equation \eqref{eq:irrep_VZ}, we obtain \eqref{eq:irrep_leeetal} holds as long as $1-\tau +  \frac{\sqrt{G}}{1+\rho\sqrt{G}}<1$, which implies that $\rho>\frac{1}{\tau}-\frac{1}{\sqrt{G}}$.	
\end{proof}

%%%%%%%%%%%%%%%%%%%%%%%%%%%%%%%%%%%%%%%%%%%%%%%%%%%%%%%%%%%%%%%%%%%%%%%%%%%%%%%%%%%%%%%%%%%%%%%%%%%%%%%%%%%%%%%%%%%%%%%%%%%%%%%%%%%%%%%%%%%%%%%%%%%%%%%%
%%%%%%%%%%%%%%%%%%%%%%%%%%%%%%%%%%%%%%%%%%%%%%%%%%%% Score function bounds %%%%%%%%%%%%%%%%%%%%%%%%%%%%%%%%%%%%%%%%%%%%%%%%%%%%%%%%%%%%%%%%%%%%%%%%%%%%%
%%%%%%%%%%%%%%%%%%%%%%%%%%%%%%%%%%%%%%%%%%%%%%%%%%%%%%%%%%%%%%%%%%%%%%%%%%%%%%%%%%%%%%%%%%%%%%%%%%%%%%%%%%%%%%%%%%%%%%%%%%%%%%%%%%%%%%%%%%%%%%%%%%%%%%%%
The next lemma establishes a bound on the dual norm of $\Omega$ of the loss gradient function under a sub-Gaussian assumption. Let $\Omega^\ast$ denote the dual norm of $\Omega$, so  $\Omega^\ast(B) =\sup \left\{\left\langle Y,B\right\rangle\left|\ Y\in C, \Omega(Y)\leq 1\right.\right\}$.  
\begin{lemma} Under Assumption \ref{assump:subgaussianity},
	\begin{equation}
		\Bbb{P}\left(\Omega^\ast(\nabla\ell(B^\star))>t\right) \leq 2N^2 \min\left\{ \exp\left( - \frac{n(1+\rho)^2t^2}{N(\sigma^2)}\right), \exp\left( - \frac{n\rho^2t^2}{\sigma^2}\right)\right\}. \label{eq:Omegastar-bound}
	\end{equation}
	\label{lemma:probabilitybound}
\end{lemma}
\begin{proof}[Proof of Lemma \ref{lemma:probabilitybound}]
	By Hoeffding's inequality for sub-Gaussian variables, for all $j,k$ and $t>0$,
	\[\prob\left(\left|\nabla_{jk}\ell\left(B^\star\right)\right|>t\right) \leq
	\prob\left(\left|\frac{1}{n}\sum_{i=1}^n\nabla_{jk}\ell_i\left(B^\star\right)\right|>t\right) \leq  2\exp\left(-n\frac{t^2}{\sigma^2}\right).\]
	Note that $(1+\rho)\sum_{i=1}^N\|B_{(i)}\|_2 \leq \Omega(B)$. Let $\Phi(B)=\frac{1}{1+\rho}\max_{i=1,\ldots N}\|B_{(i)}\|_2$. Thus,
	\begin{equation}
		\Omega^\ast(B)\leq \sup_{Y\in\real^{N\times N}}\left\{\text{Tr}(YB):\Omega_{\rho=0}(Y)\leq \frac{1}{1+\rho} \right\}=\Phi(B).\label{eq:bounddual}
	\end{equation}
	In a similar manner, $\rho\|B\|_1\leq \Omega(B)$. Setting $\Xi(B) = \frac{1}{\rho}\|B\|_\infty$, we have
	\begin{equation}
		\Omega^\ast(B)\leq \Xi(B).\label{eq:bounddual2}
	\end{equation}
	Using \eqref{eq:bounddual} and setting $\Lambda=\nabla\ell(B^\star)$,
	\begin{eqnarray*}
		\Bbb{P}\left\{\Omega^\ast(\Lambda)>t\right\} & \leq & \Bbb{P}\left\{\Phi(\Lambda) > t\right\}\nonumber\\
		& = & \Bbb{P}\left\{\max_{1\leq i \leq N}\|\Lambda_{(i)}\|_2>(1+\rho)t\right\}\\
		& \leq & \Bbb{P}\left\{\max_{1\leq i \leq N}\max_{j\neq i}|\Lambda_{ij}|>(1+\rho)\frac{t}{\sqrt{N}}\right\}\\		
		& \leq &2N(N-1)\exp\left\{-\frac{n(1+\rho)^2t^2}{2\sigma^2(N-1)}\right\}, \label{eq:probbound}
	\end{eqnarray*}
	the last inequality obtained by arguments similar to Lemma 4.3 of \cite{lee2015model}. In the same way, we can also bound the previous quantity using \eqref{eq:bounddual2} by
	\begin{eqnarray*}
		\Bbb{P}\left(\Omega^\ast(\Lambda)>t\right) & \leq & \Bbb{P}\left(\Xi(\Lambda) > t\right)\nonumber\\
		& = & \Bbb{P}\left(\|\Lambda_{(i)}\|_\infty>\rho t\right)\\
		& \leq & {N(N-1)}\exp\left(-\frac{n\rho^2t^2}{2\sigma^2}\right). \label{eq:probbound2}
	\end{eqnarray*}
	Combining \eqref{eq:probbound} and \eqref{eq:probbound2}, we can obtain equation \eqref{eq:Omegastar-bound}. 
\end{proof}
%%%%%%%%%%%%%%%%%%%%%%%%%%%%%%%%%%%%%%%%%%%%%%%%%%%%%%%%%%%%%%%%%%%%%%%%%%%%%%%%%%%%%%%%%%%%%%%%%%%%%%%%%%%%%%%%%%%%%%%%%%%%%%%%%%%%%%%%%%%%%%%%%%%%%%%%
%%%%%%%%%%%%%%%%%%%%%%%%%%%%%%%%%%%%%%%%%%%%%%%%%%%% Compatibility constants %%%%%%%%%%%%%%%%%%%%%%%%%%%%%%%%%%%%%%%%%%%%%%%%%%%%%%%%%%%%%%%%%%%%%%%%%%%
%%%%%%%%%%%%%%%%%%%%%%%%%%%%%%%%%%%%%%%%%%%%%%%%%%%%%%%%%%%%%%%%%%%%%%%%%%%%%%%%%%%%%%%%%%%%%%%%%%%%%%%%%%%%%%%%%%%%%%%%%%%%%%%%%%%%%%%%%%%%%%%%%%%%%%%%

For a semi-norm $\Psi:\real^{N\times N}\rightarrow \real$, let $\kappa_\Omega$ be the compatibility constant between $\Psi$ and the Frobenius norm, defined as
\begin{equation*}
	\kappa_\Psi = \sup\left\{\Psi(B):\|B\|_2\leq 1, B\in M\right\},
\end{equation*}
and let $\kappa_{\text{IC}}$ be the compatibility constant between the irrepresentable term and the dual norm $\Omega^\ast$ given by
\begin{equation*}
	\kappa_\text{IC} = \sup\left\{V(P_{M^\perp}(R\mathcal{H}Z-Z)) : \Omega^\ast(Z)\leq 1 \right\} \ . 
\end{equation*}

\begin{lemma} \label{lemma:compconsts}
	The following bounds on the compatibility constants hold: 
	\begin{eqnarray*}
		\kappa_\Omega & = & \sqrt{G} + \rho\sqrt{G(G-1)}, \label{eq:compat_Omega}\\
		\kappa_{\Omega^\ast} & \leq & \frac{1}{1+\rho},\label{eq:compat_Omegaast}\\
		\kappa_\text{IC} & \leq  & 3 - \tau.
	\end{eqnarray*}	
\end{lemma}

\begin{proof}[Proof of Lemma \ref{lemma:compconsts}.] Note that $\Omega(Y)$ is maximized on $\{Y:\|Y\|_2\leq 1\}$ when all entries of $Y$ have magnitude equal to $\frac{1}{\sqrt{G(G-1)}}$.   Therefore 
	\begin{equation}
		\kappa_\Omega = G\sqrt{\frac{G-1}{G(G-1)}} + \rho\frac{G(G-1)}{\sqrt{G(G-1)}} = \sqrt{G} + \rho\sqrt{G(G-1)}. 
	\end{equation}
	Similarly,  \eqref{eq:bounddual} implies
	\begin{equation}
		\kappa_{\Omega^\ast} \leq \sup\left\{\frac{1}{1+\rho}\max_{i\in\mathcal{G}}\|B_{(i)}\|_2:\|B\|_2\leq 1 \right\} \leq \frac{1}{1+\rho}.
	\end{equation}
	Finally, 
	\begin{align*}
		\kappa_{IC}  = & \sup\left\{V(P_{M^\perp}(R\mathcal{H}Z-Z)) : \Omega^\ast(Z)\leq 1 \right\}\\
		 \leq & \sup\left\{V(P_{M^\perp}(R\mathcal{H}Z)) :  \Omega^\ast(Z)\leq 1 \right\} + \sup\left\{V(P_{M^\perp}(Z)) :  \Omega^\ast(Z)\leq 1 \right\}\\
		 \leq & (1 - \tau ) + 2 = 3-\tau.
	\end{align*}
	
\end{proof}

%%%%%%%%%%%%%%%%%%%%%%%%%%%%%%%%%%%%%%%%%%%%%%%%%%%%%%%%%%%%%%%%%%%%%%%%%%%%%%%%%%%%%%%%%%%%%%%%%%%%%%%%%%%%%%%%%%%%%%%%%%%%%%%%%%%%%%%%%%%%%%%%%%%%%%%%
%%%%%%%%%%%%%%%%%%%%%%%%%%%%%%%%%%%%%%%%%%%%%%%%%%%% Main proof here %%%%%%%%%%%%%%%%%%%%%%%%%%%%%%%%%%%%%%%%%%%%%%%%%%%%%%%%%%%%%%%%%%%%%%%%%%%%%%%%%%%
%%%%%%%%%%%%%%%%%%%%%%%%%%%%%%%%%%%%%%%%%%%%%%%%%%%%%%%%%%%%%%%%%%%%%%%%%%%%%%%%%%%%%%%%%%%%%%%%%%%%%%%%%%%%%%%%%%%%%%%%%%%%%%%%%%%%%%%%%%%%%%%%%%%%%%%%
\begin{proof}[Proof of Proposition \ref{prop:error}]
	Part (a). Since $\hat{B}$ minimizes the problem \eqref{eq:main_problem}, 
	\begin{equation*}
		\ell(\hat{B}) + \lambda\Omega(\hat{B}) \leq \ell(B^\star) + \lambda\Omega(B^\star).
	\end{equation*}
	Rearranging the terms, using Assumption \ref{assump:RSC}, by the triangle inequality and the generalized Cauchy-Schwarz inequality,
	\begin{align}
		0  \geq & \ell(\hat{B}) - \ell(B^\star) + \Omega(\hat{B}) - \Omega(B^\star) \nonumber\\
		\geq & \left\langle \nabla\ell(B^\star)^T, \hat{B}-B^\star \right\rangle + \frac{m}{2}\|\hat{B}-B^\star\|_2^2 - \Omega\left(\hat{B} - B^\star\right)\nonumber\\
		 \geq & -\Omega\left(\hat{B} - B^\star\right)\Omega^\ast\left(\nabla\ell(B^\star)\right)- \Omega\left(\hat{B} - B^\star\right) + \frac{m}{2}\|\hat{B}-B^\star\|_2^2.\label{eq:proof1-lowerbound}
	\end{align}
	By the argument for computing $\kappa_\Omega$ in \eqref{eq:compat_Omega}, $\Omega(Y)\leq \{\sqrt{N} + \rho\sqrt{N(N-1)}\}\|Y\|_2$. Rearranging the terms in \eqref{eq:proof1-lowerbound},
	\begin{equation*}
		\|\hat{B}-B^\star\|_2\leq \frac{2}{m}\left\{\sqrt{N} + \rho\sqrt{N(N-1)}\right\}\left\{\lambda + \Omega^\ast\left(\hat{B} - B^\star\right)\right\}.
	\end{equation*}
	For any $\rho$, setting $\lambda = 2 \sqrt{\frac{\sigma^2\log N}{n}} \min \left\{\frac{\sqrt{N}}{1+\rho}, \frac{1}{\rho}\right\}$, by Lemma \ref{lemma:probabilitybound}, with probability at least $1-2/N$, 
	\begin{align}
		\|\hat{B}-B^\star\|_2  \leq & \frac{4}{m}\left\{\sqrt{N} + \rho\sqrt{N(N-1)}\right\}\lambda\nonumber\\
		 \leq & \frac{4}{m}\sqrt{\frac{\sigma^2\log N}{n}}\left\{\sqrt{N} + \rho N\right\}\min \left\{\frac{\sqrt{N}}{1+\rho}, \frac{1}{\rho}\right\} \nonumber\\
		 \leq & \frac{4}{m}\sqrt{\frac{\sigma^2\log N}{n}} N\min \left\{1+\rho \sqrt{N}, 1 +  \frac{1}{\rho\sqrt{N}}\right\} \nonumber\\
		 \leq & \frac{4}{m}\sqrt{N^2\frac{\sigma^2\log N}{n}}. \label{eq:errorboundlambda}
	\end{align}

	Part (b). Lemma \ref{lemma:geom_decomp} gives a geometric decomposition of the penalty. Therefore, we can directly use Theorem 3.1 of \cite{lee2015model}, since Lemma \ref{lemma:irrepresent} also ensures that their irrepresentability condition holds. Thus,
	\begin{equation*}
		\|\hat{B}-B^\star\|_2\leq \frac{2}{m}\kappa_\Omega\left(1 + \frac{\tau}{4\kappa_{\text{IC}}}\right)\lambda,
	\end{equation*}
	and $\hat{\mathcal{G}} \subseteq\mathcal{G}$ as long as 
	\begin{equation}
		\frac{4\kappa_{\text{IC}}}{\tau}\Omega^\ast\left(\nabla\ell(B^\star)\right) < \lambda < \frac{m^2\tau}{2L\kappa_\Omega^2\kappa_{\Omega^\ast}\kappa_\text{IC}}\left(1+\frac{\tau}{4\kappa_\text{IC}}\right)^{-2}.\label{eq:lambdarange}
	\end{equation}
	Setting 
	\[\lambda = \frac{8\kappa_\text{IC}}{\tau}\sqrt{\frac{\sigma^2\log N}{n}} \min \left\{\frac{\sqrt{N}}{1+\rho}, \frac{1}{\rho}\right\},\] 
	using a similar argument than \eqref{eq:errorboundlambda}, the left hand size of  \eqref{eq:lambdarange} holds with probability at least $1-2/N$. The right hand side of \eqref{eq:lambdarange} holds as long as the sample size satisfies
	\begin{equation*}
		n > C(L, m, \tau, \kappa_\Omega, \kappa_{\Omega^\ast}, \kappa_{IC}) \left(\sqrt{G} + G\right)^2\sigma^2\log N,
	\end{equation*}
	with $C(L, m, \tau, \kappa_\Omega, \kappa_{\Omega^\ast}, \kappa_{IC})>0$ a positive constant. Therefore, claims \eqref{eq:frobenius_rate_Irrep} and \eqref{eq:support_consistency} follow. 
	
\end{proof}

\section{Data aquisition and pre-processing}
\label{appendix:data}

\subsection{Subjects and imaging}

\subsubsection*{The COBRE data}

Raw anatomic and functional scans from 146 subjects (72 psychosis patients and 74 healthy control subjects) were downloaded from a public database (\url{http://fcon_1000.projects.nitrc.org/indi/retro/cobre.html}). Four subjects coded as ambidextrous (2 patients, 2 controls) were excluded to yield 70 psychosis patients and 72 controls for analysis. To enter the COBRE dataset, subjects had a diagnosis of either schizophrenia or schizoaffective disorder and were without histories of neurological disorder, mental retardation, severe head trauma with more than 5 minutes loss of consciousness and substance abuse/dependence within the last 12 months.

In the primary sample, two schizophrenic (SCZ) subjects and one healthy control (HC) subject had insufficient voxels in the cerebrospinal fluid (CSF) segmentation on the CSF, and they were dropped from additional analyses. Two additional SCZ subjects were excluded for scrub ratios (see discussion of \emph{scrubbing routine} in fMRI Data Analysis) greater than 0.6, leaving 38 SCZ subjects and 42 HC subjects for the analysis. In the replication sample, 15 psychosis patients and two control subjects were excluded for scrub ratios greater than 0.6; one patient was excluded with incomplete data,  leaving 54 SCZ and 69 HC subjects for analysis (see Table \ref{table:patients}). 

A full description of the imaging parameters for the COBRE dataset is available online at the link provided above and in  \cite{Aine2017}.

\subsubsection*{The UMich data}

Subjects were selected from experiments conducted by Professor Stephan F. Taylor at the University of Michigan between 2004 and 2011 for task-based fMRI studies, which included resting state scans. Thirty-nine stable outpatients were selected with DSM-IV schizophrenia or schizoaffective disorder (SCZ) \citep{american1994diagnostic}. Forty healthy comparison (HC) subjects, without a lifetime history of Axis I psychiatric disorders \citep{first1995structured}, were selected to approximate the age range, gender distribution and family education level of the patients. Prior to initial data collection, all subjects gave written, informed consent to participate in the protocol approved by the University of Michigan institutional review board (IRBMED).

MRI scanning occurred on a GE 3T Signa scanner (LX [8.3] release, General Electric Healthcare, Buckinghamshire, United Kingdom). Functional images were acquired with a T2*-weighted, reverse spiral acquisition sequence (gradient recalled echo, TE=30 msec, FA=90 degrees, field of view=22 cm, 40 slices, 3.0mm thick/0mm skip, equivalent to 64 x 64 voxel grid – yielding isotropic voxels 3 mm on edge). Because the data were acquired across different experiments, acquisition parameters differed slightly in the aggregate sample: 240 volumes at TR=1500 msec (11 SCZ, 10 HC), 180 volumes at TR=2000 msec (17 SCZ, 16 HC) and 240 volumes at 2000 msec (14 SCZ, 17 HC). Acquisitions were acquired in the resting state with eyes open and fixated on a large ‘plus’ sign projected on a monitor.

\begin{table}[!htbp]
	\centering
	{\footnotesize
	\begin{tabular}{llllll}
		\hline 
		\textbf{Dataset}& \textbf{\# nodes} & \textbf{Status}  & \textbf{\#} & \textbf{Sex M/F} & \textbf{Age mean (s.d.)} \\ 
		\hline 
		COBRE & 263 & Schizophrenic & 54 & 48/6  & 35.4 (13.1) \\ 
		& & Control & 70 & 48/22 & 35.1 (11.5) \\ 
		\hline 
		UMich & 264 &  Schizophrenic & 39 & 29/10 & 40.7 (11.5) \\ 
		& & Control & 40 & 28/12  & 36.8 (12.3) \\ 
		\hline 
	\end{tabular} }
	\caption{Summary of the two datasets.}
	\label{table:patients}
\end{table}

\subsection{Pre-processing}

We first performed standard pre-processing steps. All scans were slice-time corrected and realigned to the 10th image acquired during a scanning session \citep{jenkinson2002improved}. Subsequent processing was performed with the Statistical Parametric Mapping SPM8 package (Wellcome Institute of Cognitive Neurology, London). Anatomic normalization was done with the VBM8 toolbox in SPM8, using the high resolution structural scans obtained for both datasets. Normalizing warps were applied to the co-registered, functional volumes, which were re-sliced and smoothed with an 8 mm isotropic Gaussian smoothing kernel. To assess and manage movement, we calculated the frame-wise displacement (FD) \citep{power2012spurious}, for all 6 parameters of rotation and translation. We used a \textit{scrubbing routine} to censor any frame with FD $>$ 0.5 mm from the regression analysis described below, yielding a scrub ratio for each subject. Three-compartment segmentation of the high-resolution structural image from the VBM8 normalization was applied to the functional time series to extract cerebral spinal volume (CSF) and white matter (WM) compartments, which were then subjected to a principal component analysis to identify the top 5 components in each \citep{behzadi2007component}, which should correspond to heart rate and respiratory effects on global signal \citep{chai2012anticorrelations}. Multiple regressions were applied to the time series to remove the following nuisance effects: linear trend, 6 motion parameters, their temporal derivatives, the quadratics of these 12 parameters, 5 components from the PCA of CSF, 5 components of PCA of WM, followed by band pass filtering from 0.01 – 0.1 Hz, and then motion scrubbing. For each 4D data set, time courses were then extracted from 10 mm diameter spheres based on the 264 sets of coordinates identified by \cite{power2011functional}. From these time series, a cross-correlation matrix of Pearson r-values was obtained and Fisher's R-to-Z transformation was applied for each of the 264 nodes with every other node (for COBRE dataset, node 75 is missing). Finally, for each individual, edge weights were assigned to be ranks of these score, with edge scores ranked separately for each subject, and then these values were centered and standardized across the individuals. Ranks have been used previously in brain connectomic studies to reduce the effect of potential outliers \citep{yan2013standardizing}; although some information is lost with the rank transformation, we observed that while ranks do not increase the classification accuracy significantly, they tend to produce sparser solutions with a similar accuracy to Pearson correlations.

\end{document}